\documentclass{article}

\usepackage[main, final]{neurips_2026}

\usepackage[utf8]{inputenc} 
\usepackage[T1]{fontenc}    
\usepackage{hyperref}       
\usepackage{url}            
\usepackage{booktabs}       
\usepackage{graphicx}       
\usepackage{amsmath}        
\usepackage{amssymb}        
\usepackage{amsfonts}       
\usepackage{amsthm}         
\usepackage{nicefrac}       
\usepackage{microtype}      
\usepackage{xcolor}         

\usepackage{multirow}
\usepackage{caption}
\usepackage{makecell}
\usepackage[normalem]{ulem}
\usepackage{wrapfig}
\usepackage{placeins}
\usepackage{comment}

\newtheorem{proposition}{Proposition}
\title{PocketVE: Stable and Property-Guided Structure-Based Drug Design with Variance-Exploding Diffusion}

\author{%
  Peining Zhang \\
  University of Connecticut\\
  Storrs, CT 06269, USA\\
  \texttt{peining.zhang@uconn.edu}
  \And
  Jinbo Bi \\
  University of Connecticut\\
  Storrs, CT 06269, USA\\
  \texttt{jinbo.bi@uconn.edu}
}

\begin{document}

\maketitle

\begin{abstract}
  Protein-conditioned 3D molecule generation is a central challenge in structure-based drug design, requiring a balance between pocket compatibility, molecular properties, and physical geometry.
  We propose \textbf{PocketVE}, a protein-pocket-conditioned variance-exploding (VE) diffusion framework that couples stable coordinate denoising with inference-time property guidance.
  Specifically, PocketVE combines an EDM-style training and sampling setup for 3D denoising, classifier-free guidance for multi-property steering, and adaptive protein perturbation as a training-time pocket regularizer.
  On the CrossDocked2020 benchmark under GenBench3D, PocketVE improves Valid$_{3\text{D}}$ from 58.6 to 80.6 and reduces strain energy from 457.4 to 127.9 relative to its guided TAGMol architectural parent; relative to TargetDiff, it attains comparable Valid$_{3\text{D}}$ with lower strain energy (127.9 vs.\ 306.0), while retaining competitive docking and molecular-property scores under moderate guidance.
  A guidance-scale study shows that moderate guidance gives a favorable balance between target-related objectives and geometric quality, whereas stronger guidance can degrade geometry and distributional fidelity.
  Pocket-permutation and PoseCheck diagnostics further support pocket-specific spatial compatibility with reduced steric conflicts.
  Overall, the results suggest that geometric stability and inference-time property guidance should be considered as coupled design objectives.
\end{abstract}
\section{Introduction}
\label{sec:introduction}
Generating drug-like molecules with desirable properties and high affinity to a given protein binding site, a task known as structure-based drug design (SBDD), sits at the intersection of 3D geometry and molecular property optimization.
A generated ligand must satisfy two classes of constraints: pocket-dependent ones (shape complementarity, binding pose, steric fit) and intrinsic molecular ones (drug-likeness, synthetic accessibility, logP)~\cite{isert2023structure,zhang2025unraveling}.
Generative models for molecular design have progressively transitioned from 1D strings~\cite{gomez2018automatic,loeffler2024reinvent} and 2D graphs~\cite{jin2018junction,shigraphaf} to direct 3D modeling~\cite{luo20213d,peng2022pocket2mol}.
Among 3D generative approaches, non-autoregressive diffusion models have substantially advanced target-aware generation~\cite{ho2020denoising,songscore,hoogeboom2022equivariant}: TargetDiff~\cite{guan3d} introduced SE(3)-equivariant diffusion to jointly denoise the full ligand without relying on a fixed generation order, while TAGMol~\cite{dorna24tagmol} demonstrated that gradient-guided sampling can steer the generative distribution toward regions of higher binding affinity.

Despite these advances, integrating strong protein-ligand interactions without deteriorating molecular geometry remains an open challenge.
First, the reverse sampling trajectory can be sensitive to the conditioning signal, creating a trade-off between steerability and stability.
While image models often work with bounded pixels or latent representations~\cite{rombach2022high}, where strong guidance can lead to overexposed outputs~\cite{zhang2024tackling,lou2023reflected}, guidance in SBDD acts directly on unconstrained Euclidean coordinates.
Consequently, excessive guidance pressure often drives the geometry toward non-physical conformations, as reflected by the atom clashes and high strain energy reported in recent GenBench3D evaluations~\cite{baillif2024benchmarking}, where higher strain indicates a less physically plausible 3D conformation.
Second, existing models primarily focus on learning the unconditional chemical distribution of the training data, leaving the integration of real-world multi-objective properties highly inflexible and poorly transferable~\cite{dorna24tagmol,jian2026general,dhariwal2021diffusion}.
As a result, adjusting how strongly a trained model pursues the protein target or a specific drug-like profile usually requires training a separate time-dependent classifier for every property of interest, making it poorly scalable when balancing multiple objectives.

Motivated by these challenges, this paper studies how conditional control changes the property--geometry trade-off in protein-conditioned 3D generation.
Our main point is that the guidance mechanism and coordinate backbone should be designed jointly, because both shape the reverse trajectory in Euclidean space.
In particular, a VE-style coordinate backbone provides a natural basis for stable denoising across noise scales, while pocket-aware scaling and sampling choices determine how that backbone interacts with the protein coordinate frame.
We therefore evaluate PocketVE as an integrated target-aware framework, without attributing its overall improvement to VE alone.
PocketVE combines VE/EDM-style coordinate parameterization, pocket-aware scale preservation, classifier-free guidance for multi-property steering, and training-time pocket perturbation.
Figure~\ref{fig:overview} summarizes the overall pipeline and the three design axes of PocketVE.

\begin{figure}[t]

\centering
\includegraphics[width=0.99\textwidth,trim=80pt 45pt 95pt 55pt,clip]{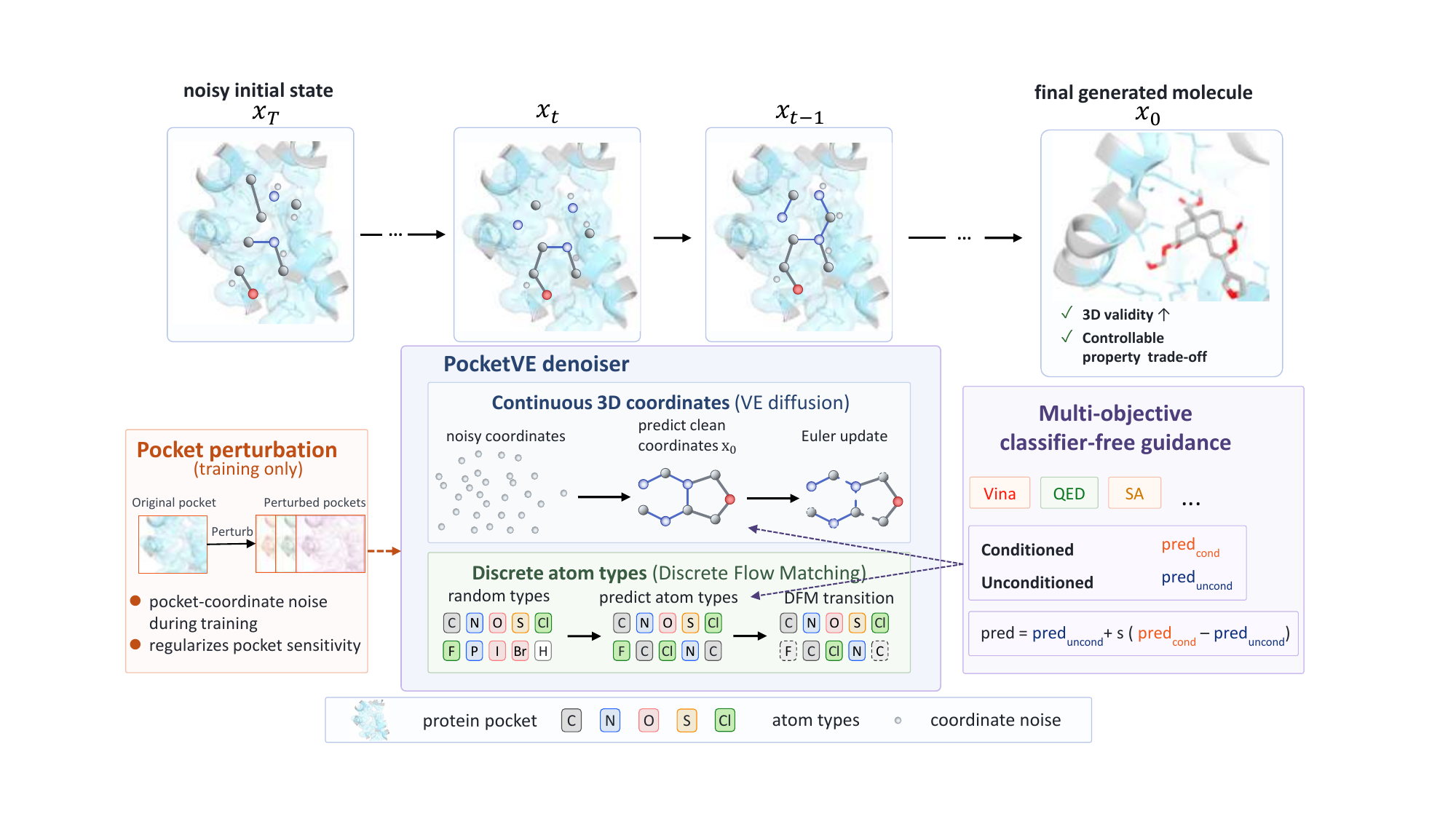}

\captionsetup{skip=4pt}
\caption{Overview of PocketVE.
PocketVE couples VE-based coordinate denoising with discrete flow matching for atom types.
Multi-objective classifier-free guidance provides inference-time property steering, while protein-pocket perturbation regularizes the denoiser during training.}
\label{fig:overview}

\end{figure}

We evaluate our framework on the CrossDocked2020 benchmark under the GenBench3D protocol~\cite{baillif2024benchmarking}.
Relative to guided TAGMol, PocketVE increases Valid$_{3\text{D}}$ from 58.6 to 80.6 and reduces strain energy from 457.4 to 127.9, while retaining competitive docking and molecular-property scores under moderate guidance.
TargetDiff and PAFlow provide complementary geometry and affinity references: PocketVE is comparable to TargetDiff in Valid$_{3\text{D}}$ and has lower strain, whereas PAFlow obtains stronger Vina scores at a substantial geometric cost.
Guidance experiments further show a trade-off between target-related properties and geometric fidelity.
Together, these findings suggest that geometric stability and conditional steering should be considered jointly in target-aware 3D molecular generation.

Our main contributions are as follows:
\begin{itemize}
  \item We develop a pocket-aware VE coordinate parameterization~\cite{karras2022elucidating} that preserves the shared protein--ligand coordinate frame while using EDM-style preconditioning for stable denoising, together with discrete flow matching for atom types.
  \item We characterize the guidance-scale trade-off in protein-conditioned 3D molecule generation, showing that Vina Score, QED, and SA respond differently from geometry metrics such as Valid$_{3\text{D}}$ and strain energy.
  \item We combine this backbone with multi-objective CFG and protein-side perturbation, improving the balance between structural validity, drug-like properties, and target-specific affinity on CrossDocked2020 under the GenBench3D protocol.
\end{itemize}

\section{Related Work}
\label{sec:related_work}
\paragraph{Target-Agnostic 3D Molecular Generation.}
In target-agnostic 3D molecular generation, Equivariant Diffusion Models learn coordinate denoising with E(3)-equivariant networks~\cite{hoogeboom2022equivariant}. 
Subsequent works extend this paradigm to latent representations and joint structure--geometry modeling~\cite{xu2023geometric,vignac2023midi}, as well as unified generation of atom types, coordinates, and bonds via flow matching~\cite{irwin2025semlaflow}.
More recently, VE-style parameterizations further improve geometric stability in 3D generation~\cite{zhang2026veda}. 
These advances suggest that geometric stability should be built into the generative backbone rather than enforced post hoc.
PocketVE brings this design principle into the protein-conditioned setting, where ligand coordinates must remain consistent with an explicit pocket geometry.

\paragraph{Structure-Based Drug Design.}
Early deep generative models for structure-based drug design (SBDD), such as 3D-SBDD~\cite{luo20213d} and Pocket2Mol~\cite{peng2022pocket2mol}, established spatial representations and architectures for pocket-conditioned generation.
However, many of these methods rely on autoregressive atom or motif placement, which is prone to error accumulation along a fixed generation order.
Non-autoregressive diffusion methods instead denoise all ligand atoms jointly.
TargetDiff~\cite{guan3d} introduced an SE(3)-equivariant formulation for target-aware 3D diffusion, while DiffSBDD~\cite{schneuing2024structure} further developed equivariant diffusion models for structure-based generation.
IPDiff~\cite{huang2024protein} incorporates protein--ligand interaction priors, BindDM adaptively extracts interaction-relevant protein--ligand subcomplexes~\cite{huang2024binding}, PAFlow~\cite{zhou2025prior} adopts prior-guided flow matching, and PocketXMol studies direct clean-coordinate prediction in pocket-conditioned molecular design~\cite{peng2026unified}.
Flow-based formulations are attractive for sampling efficiency, but recent evidence suggests that low-step generation alone does not guarantee chemically reliable 3D geometry~\cite{nikitin2025geom}.
PoseCheck~\cite{harris2023benchmarking} analyzes generated protein--ligand poses directly and shows that physical violations or missing key interactions may be obscured when evaluation relies on redocking.
GenBench3D~\cite{baillif2024benchmarking} further emphasizes the 3D conformation quality of generated ligands.
For example, it evaluates whether generated bond lengths and valence angles are consistent with reference distributions from 3D structure databases such as CSD~\cite{groom2016cambridge} and LigBoundConf~\cite{tong2021large}.
Together, these findings motivate our focus on the property--geometry trade-off: Vina-based docking scores and molecular properties should be interpreted together with pose and conformation quality.

\paragraph{Conditional Diffusion.}
Mainstream strategies for injecting conditional signals modify the reverse sampling trajectory via classifier guidance~\cite{dhariwal2021diffusion} or classifier-free guidance (CFG)~\cite{ho2022classifier}.
However, applying these paradigms to 3D molecular design is non-trivial: stronger conditional pressure can degrade physical geometry, causing atom clashes or pushing molecules off-distribution~\cite{baillif2024benchmarking}.
TAGMol~\cite{dorna24tagmol} uses external classifier guidance to steer sampling toward desired molecular properties, but this requires separate property predictors and scale calibration.
BADGER~\cite{jian2026general} further studies guidance in diffusion-based SBDD, combining classifier guidance and classifier-free guidance for binding-affinity control and extending the same framework to joint optimization over affinity, QED, and SA.
Other approaches, including alignment methods such as AliDiff~\cite{gu2024aligning} and reinforcement-learning-based guidance~\cite{zhou2025guiding}, directly optimize or fine-tune the generator toward preferred properties.
In particular, we study this trade-off under a VE-style 3D backbone, where stronger conditional control may improve target objectives while introducing geometric drift.

\section{Method}
\label{sec:method}
\subsection{Overview}
\label{subsec:method_overview}
We model the denoising process over a hybrid state space that contains continuous coordinates and discrete atom types.
The task of structure-based drug design (SBDD) is to learn a conditional distribution $p(\mathcal{L} \mid \mathcal{P})$ over ligands conditioned on a protein pocket.
A ligand with $N_L$ atoms is represented as $\mathcal{L}=(\mathbf{x},\mathbf{z})$, where $\mathbf{x}\in\mathbb{R}^{N_L\times 3}$ are atom coordinates and $\mathbf{z}\in\{0,1\}^{N_L\times S}$ are atom-type indicators.
The protein pocket is written as $\mathcal{P}=(\mathbf{y},\mathbf{c})$, where $\mathbf{y}\in\mathbb{R}^{M\times 3}$ are pocket atom coordinates and $\mathbf{c}$ are pocket atom features. Our goal is to model $p_\theta(\mathcal{L}\mid \mathcal{P})$, so that the generated ligand is both geometrically valid and compatible with the target pocket.

Our framework comprises three coupled components.
It uses an EDM-style~\cite{karras2022elucidating} generation backbone to stabilize continuous coordinate denoising, classifier-free guidance to control auxiliary property targets at inference time, and training-time protein perturbation to reduce over-reliance on a single rigid pocket realization.

\subsection{EDM-style generation backbone}
\label{sec:ve-paradigm}
For the continuous coordinates, we use a variance-exploding (VE) forward process. 
Given a clean ligand coordinate set $\mathbf{x}_0$, we sample a noise level $\sigma \sim p(\sigma)$ and form
\begin{equation}\label{eq:x_sigma}
  \mathbf{x}_\sigma = \mathbf{x}_0 + \sigma \boldsymbol{\epsilon}, \qquad \boldsymbol{\epsilon}\sim \mathcal{N}(\mathbf{0}, \mathbf{I}).
\end{equation}
In practice, the noise scale $\sigma$ is sampled from a log-normal distribution:
$\sigma \sim \operatorname{LogNormal}\!\left(\ln\sqrt{\sigma_{\min}\sigma_{\max}}, \left[\frac{1}{8}\ln(\sigma_{\max} / \sigma_{\min})\right]^2\right)$,
so $\sigma$ lies near the range $[\sigma_{\min}, \sigma_{\max}]$ in most cases and follows the usual VE/EDM-style noise schedule~\cite{karras2022elucidating}.

For atom types, we use a separate discrete corruption process.
Following the discrete flow matching view~\cite{campbell2024generative}, this branch learns a categorical denoising target over the same hybrid state space.
We use the corruption kernel
\begin{equation}\label{eq:z_corrupt}
  q(\mathbf{z}_\sigma \mid \mathbf{z}_0) = (1-m(\sigma))\,\delta(\mathbf{z}_\sigma=\mathbf{z}_0) + m(\sigma)\,\mathrm{Unif}(K),
\end{equation}
where $\mathrm{Unif}(K)$ is the uniform distribution over atom types.
The masking schedule is $m(\sigma) = \frac{\ln \sigma - \ln \sigma_{\min}}{\ln \sigma_{\max} - \ln \sigma_{\min}}$.
This keeps the continuous coordinate diffusion and the discrete type diffusion aligned at the same noise level.

A small but important target-aware detail is the input scaling used before the equivariant backbone.
In standard EDM, the noisy sample is often scaled as
\(
\mathbf{x}_\sigma / \sqrt{\sigma^2+\sigma_{\mathrm{data}}^2}
\)
before entering the backbone.
In our setting, however, ligand coordinates interact with unscaled protein-pocket coordinates inside the same distance-based geometric graph.
Directly using the standard factor would shrink low-noise ligand coordinates by $1/\sigma_{\mathrm{data}}$, while leaving protein coordinates unchanged, thereby distorting protein--ligand distances.
We therefore use the scale-preserving input $\tilde{\mathbf{x}}_\sigma
=
\frac{\sigma_{\mathrm{data}}}{\sqrt{\sigma^2+\sigma_{\mathrm{data}}^2}}
\mathbf{x}_\sigma$,
which recovers the original ligand scale as $\sigma\to 0$ and still attenuates high-noise inputs.
This preserves the physical coordinate frame shared by ligand and pocket atoms.
We use $\sigma_{\mathrm{data}}=10.0$ as a fixed coordinate-scale parameter, rather than estimating it as a per-ligand sample variance or normalizing each ligand independently.
This value covers both the internal ligand extent and the ligand displacement from the pocket center under pocket-centered normalization.

For the coordinate branch, we then use the EDM preconditioning form
\begin{equation}\label{eq:edm_denoiser}
  D_\theta(\mathbf{x}_\sigma, \mathbf{z}_\sigma, \sigma, \mathcal{P})
  =
  c_{\mathrm{skip}}(\sigma)\mathbf{x}_\sigma
  + c_{\mathrm{out}}(\sigma)\,
  F_\theta\!\left(\tilde{\mathbf{x}}_\sigma;\, c_{\mathrm{noise}}(\sigma),\, \mathbf{z}_\sigma,\, \mathcal{P}\right).
\end{equation}
Here $F_\theta$ is the equivariant backbone shared by both branches, and the skip term keeps the coordinate prediction anchored to the noisy input.
This plays the same role as the preconditioning in VEDA~\cite{zhang2026veda}, but here it is conditioned on the pocket.
Further coefficient details are provided in Appendix~\ref{appendix:edm_preconditioning}.

Atom types are recovered through a categorical prediction module built on top of the same equivariant features.
Its output
\(
H_\theta(\mathbf{x}_\sigma, \mathbf{z}_\sigma, \sigma, \mathcal{P}) \in \mathbb{R}^{N_L\times K}
\)
gives the logits for the clean atom types.
The clean ligand is then recovered by
\begin{equation}\label{eq:clean_pred}
  \hat{\mathbf{x}}_0 = D_\theta(\mathbf{x}_\sigma, \mathbf{z}_\sigma, \sigma, \mathcal{P}), \qquad
  \hat{\mathbf{z}}_0 = \mathrm{softmax}\!\big(H_\theta(\mathbf{x}_\sigma, \mathbf{z}_\sigma, \sigma, \mathcal{P})\big).
\end{equation}

The training objective is a joint denoising loss over geometry and atom identity:
\begin{equation}\label{eq:base_loss}
  \mathcal{L}_{\mathrm{base}}(\theta)
  =
  \mathbb{E}_{\mathcal{L}_0,\mathcal{P},\sigma}\bigg[
    \lambda_x(\sigma)\,\lVert \hat{\mathbf{x}}_0 - \mathbf{x}_0 \rVert_2^2
    + \lambda_z(\sigma)\,\mathrm{CE}\!\left(\hat{\mathbf{z}}_0,\mathbf{z}_0\right)
    \bigg].
\end{equation}

\subsection{Sampling procedure}
\label{subsec:sampling_procedure}
At inference time, we follow a fixed decreasing VE noise schedule.
At each step, the coordinate branch predicts the clean structure and applies a first-order Euler update:
\begin{equation}\label{eq:sampling_coord_update}
\begin{aligned}
\mathbf{d}_i
&=
\frac{\mathbf{x}_{\sigma_i}-D_\theta(\mathbf{x}_{\sigma_i},\mathbf{z}_{\sigma_i},\sigma_i,\mathcal{P})}{\sigma_i},
\qquad
\mathbf{x}_{\sigma_{i+1}}
=
\mathbf{x}_{\sigma_i}
+
(\sigma_{i+1}-\sigma_i)\mathbf{d}_i.
\end{aligned}
\end{equation}
Ligand coordinates are initialized around the pocket center, and ligand atom types are initialized from the uniform discrete prior.
For the discrete atom types, we update $\mathbf{z}_{\sigma_i}$ with a DFM-based discrete sampler driven by the categorical predictor $H_\theta$.
The number of ligand atoms is sampled from a prior conditioned on the estimated pocket size, matching the inference-time size prior used in TAGMol~\cite{dorna24tagmol}.
When classifier-free guidance is enabled, we use the same sampler with a tunable guidance weight.
The concrete noise range, generalized arcsin schedule, discrete-sampler details, and sampling-step ablations are provided in Appendix~\ref{appendix:sampling_details} and Appendix~\ref{appendix:sampling_steps}.

\subsection{Classifier-free guidance}
\label{subsec:cfg}
We focus on three commonly used generation-quality objectives: Vina for predicted binding affinity, QED for drug-likeness, and normalized SA for synthetic accessibility.
To make generation steerable, we discretize each target property into percentile bins computed on the training set and encode the resulting bin indices as an auxiliary condition $\mathbf{a}$.
Multi-objective control at inference time is specified by choosing one target bin per property and concatenating the resulting indices into the same auxiliary condition vector.
We apply condition dropout~\cite{ho2022classifier} to this auxiliary input during training while always keeping the pocket condition $\mathcal{P}$.
For notational simplicity, this auxiliary input was suppressed in the backbone definition above and is written explicitly only in this subsection.
Let $\varnothing$ denote the dropped auxiliary condition.
Following the standard CFG notation~\cite{ho2022classifier}, let $w$ denote the extrapolation weight between the conditional and dropped-condition branches.
In our implementation, however, we use the common sampling-time parameterization
\(
s = w + 1
\),
so that $s=0$ corresponds to the null / unconditional branch, $s=1$ recovers the plain conditional model, and $s>1$ produces the usual CFG extrapolation.
For comparison, an explicit classifier-guidance rule with the same scalar coefficient $w$ appears in the VE sampler update as
\begin{equation}\label{eq:classifier_guidance_ode}
\mathbf{d}_{\mathrm{guide}}
=
\mathbf{d}_{\mathrm{base}}
- w \sigma \nabla_{\mathbf{x}_\sigma}\log p_\phi(\mathbf{a}\mid \mathbf{x}_\sigma,\mathbf{z}_\sigma,\mathcal{P}).
\end{equation}
Here $\mathbf{d}_{\mathrm{guide}}$ and $\mathbf{d}_{\mathrm{base}}$ denote the guided and unguided VE update fields, respectively.
Thus, when $w$ is used as an external guidance coefficient, the property gradient enters the sampler through a $\sigma$ factor, making its effect noise-level dependent.
Appendix~\ref{appendix:cfg_noise_dependence} gives the corresponding score-level derivation and the equivalent $x_0$-space form.
In addition, explicit classifier guidance requires objective-specific calibration of both scale and direction.
Different property targets can live on very different numeric ranges, such as QED in $[0,1]$ versus Vina scores on a much wider scale, and some quantities are maximized while others, such as docking scores, are minimized.
As a result, when $w$ is used as an external guidance coefficient, its value and sign are generally not comparable across different property heads.

In contrast, CFG acts directly on the model predictions, which reduces the need for noise-level-dependent external-gradient calibration across different property metrics.
Therefore, CFG induces the same linear extrapolation in the $x_0$ prediction, but without introducing a separate property predictor.
For atom types, we apply the same CFG rule in the logit space of the categorical predictor, which is the standard practical counterpart of CFG for discrete outputs.
Under this parameterization, this becomes
\begin{equation}\label{eq:cfg_prediction}
\begin{aligned}
  \hat{\mathbf{x}}_\theta^{\mathrm{cfg}}
  &=
  D_\theta(\mathbf{x}_\sigma,\mathbf{z}_\sigma,\sigma,\mathcal{P},\varnothing)
  + s\!\left(
  D_\theta(\mathbf{x}_\sigma,\mathbf{z}_\sigma,\sigma,\mathcal{P},\mathbf{a})
  - D_\theta(\mathbf{x}_\sigma,\mathbf{z}_\sigma,\sigma,\mathcal{P},\varnothing)\right),\\
  \hat{\mathbf{z}}_\theta^{\mathrm{cfg}}
  &=
  \mathrm{softmax}\!\Big(H_\theta(\mathbf{x}_\sigma,\mathbf{z}_\sigma,\sigma,\mathcal{P},\varnothing)
  + s\!\left(
  H_\theta(\mathbf{x}_\sigma,\mathbf{z}_\sigma,\sigma,\mathcal{P},\mathbf{a})
  - H_\theta(\mathbf{x}_\sigma,\mathbf{z}_\sigma,\sigma,\mathcal{P},\varnothing)\right)\Big).
\end{aligned}
\end{equation}
Here the first line applies CFG to the coordinate prediction, while the second applies the same extrapolation to the type-logit branch before the softmax.
We provide the corresponding score-level derivation and its relation to standard CFG formulas in Appendix~\ref{appendix:cfg_recover}.
Larger $s$ makes the model follow the auxiliary property condition more strongly, while smaller $s$ keeps the sample closer to the property-unconditional distribution under the same pocket.

\subsection{Adaptive protein perturbation}
\label{subsec:protein_perturbation}

The protein pocket is not treated as a perfectly clean ground truth during training.
Experimental protein structures are subject to measurement error and capture only a static snapshot of a dynamic conformational ensemble.
To account for this, we apply a small perturbation operator
\begin{equation}\label{eq:pocket_perturb}
  \tilde{\mathcal{P}} = \mathrm{Perturb}(\mathcal{P}; \sigma_p),
\end{equation}
where $\sigma_p$ controls the scale of Gaussian coordinate perturbation on the pocket atoms.
We train the denoiser on $\tilde{\mathcal{P}}$ instead of $\mathcal{P}$.
This perturbation regularizes against over-reliance on a single crystallographic structure and encourages a neighborhood-aware conditional distribution.
Following the classical connection between input noise and smoothness regularization~\cite{bishop1995training,chapelle2000vicinal}, a second-order expansion of a smooth denoising loss around the clean pocket shows that Gaussian pocket perturbation adds a leading-order penalty on the denoiser's sensitivity to pocket coordinates.
For a squared denoising loss, this leading term contains $\frac{\sigma_p^2}{2}\lVert J_{\mathcal{P}} f_\theta\rVert_F^2$, where $J_{\mathcal{P}} f_\theta$ is the Jacobian of the denoising prediction with respect to the pocket input.
Thus, the perturbation discourages sharp changes in the coordinate and type predictions under small pocket-coordinate variations, while still preserving the pocket as the conditioning signal.
For the main results, we use a bounded noise-level-dependent schedule $\sigma_p=\min(0.1\sigma,0.5)$.
A formal statement and derivation are provided in Appendix~\ref{appendix:pocket_perturbation_details}.

\subsection{Full objective}
\label{subsec:full_objective}

The final training objective is
\begin{equation}\label{eq:full_loss}
  \mathcal{L}(\theta)
  =
  \mathbb{E}_{\mathcal{L}_0,\mathcal{P},\sigma,\tilde{\mathcal{P}}}
  \bigg[
    \lambda_x(\sigma) \lVert \hat{\mathbf{x}}_0^{\tilde{\mathcal{P}}} - \mathbf{x}_0 \rVert_2^2
    - \lambda_z(\sigma) \sum_{i=1}^{N_L} \sum_{k=1}^K (\mathbf{z}_0)_{ik} \log (\hat{\mathbf{z}}_0^{\tilde{\mathcal{P}}})_{ik}
    \bigg],
\end{equation}
where $\hat{\mathbf{x}}_0^{\tilde{\mathcal{P}}}$ and $\hat{\mathbf{z}}_0^{\tilde{\mathcal{P}}}$ denote the predictions calculated using the perturbed pocket $\tilde{\mathcal{P}}$. 
At test time, the same model can be run either unguided or with classifier-free guidance, which gives a steerable trade-off between fidelity and target-specific preference.
\section{Experiments}
\label{sec:experiments}
\subsection{Experimental setup}
\label{subsec:experimental_setup}
\paragraph{Dataset}
We train our models on the CrossDocked2020 dataset~\cite{francoeur2020three}.
Following prior work~\cite{luo20213d,peng2022pocket2mol,guan3d}, we filter out binding poses with an RMSD $> 1$ \AA{} and remove protein pairs with sequence identity $> 30\%$, resulting in a high-quality subset common in target-aware 3D molecule generation.
This split is the standard benchmark for target-aware 3D generation, but docking-centered evaluation alone does not fully characterize whether generated ligands are geometrically plausible.
To address this, we further adopt the GenBench3D protocol~\cite{baillif2024benchmarking}, which is designed to assess 3D conformation quality and check whether molecules remain physically plausible under stronger conditional guidance.

\paragraph{Baselines}
We compare \textbf{PocketVE} against representative target-aware 3D molecule generation baselines, including AR~\cite{luo20213d}, Pocket2Mol~\cite{peng2022pocket2mol}, TargetDiff~\cite{guan3d}, DecompDiff~\cite{guan2023decompdiff}, IPDiff~\cite{huang2024protein}, PAFlow~\cite{zhou2025prior}, SeFMol~\cite{zhang2026steering}, ALiDiff~\cite{gu2024aligning}, PocketXMol~\cite{peng2026unified}, and BindDM~\cite{huang2024binding}.
\textbf{TAGMol}~\cite{dorna24tagmol} is the architectural parent and ablation reference because it shares the base equivariant denoiser with PocketVE; this comparison reflects a bundled system transition rather than an isolated VE effect.
We use TargetDiff as the primary established geometry comparison and PAFlow as the strong-affinity comparison that exposes the affinity--geometry trade-off.
For baselines with released generated molecules, we re-evaluate those outputs under the same test split, docking protocol, and GenBench3D pipeline used for PocketVE.

\paragraph{Metrics}
We evaluate generated molecules on 100 test proteins, sampling 100 molecules per protein, and report both mean and median results.
For binding affinity, we use AutoDock Vina~\cite{eberhardt2021autodock} under the common setup of Luo et al.~\cite{luo20213d} and Ragoza et al.~\cite{ragoza2022generating}, reporting Vina Score, Vina Min, and Vina Dock.
We also report High Affinity, the per-pocket fraction of generated molecules whose Vina Dock score is no worse than the reference ligand, and Joint-Ref, the fraction that simultaneously satisfies no-worse Vina Dock, QED, and SA than the reference ligand.
For molecular properties, we report QED~\cite{bickerton2012quantifying}, normalized SA~\cite{ertl2009estimation}, and diversity.
For geometry, we report Valid$_{3\text{D}}$, clash-free rate, strain energy, and centroid distance under the GenBench3D protocol~\cite{baillif2024benchmarking}.
We also include TargetDiff-style distribution diagnostics based on Jensen-Shannon Divergence (JSD) over empirical ligand distance and atom-type statistics~\cite{guan3d}.
These JSD metrics are complementary distribution-fidelity diagnostics rather than explicit physical-validity metrics.

\paragraph{Implementation details}
PocketVE uses the target-aware equivariant network architecture of TAGMol~\cite{dorna24tagmol} as the base denoiser, while modifying the diffusion parameterization, guidance mechanism, sampling procedure, and optimizer.
For conditional control, we use training-set labels for Vina, QED, and normalized SA, discretize each property independently into five percentile bins, and apply joint condition dropout with probability $0.5$ during training.
In the main conditional experiments, the sampling condition uses the most favorable target bin for each property, namely the lowest Vina bin and the highest QED and SA bins.
We train for $400{,}000$ optimization steps with batch size $4$, validate every $2{,}000$ steps, and report the checkpoint with the lowest validation loss.
For optimization, we use Muon~\cite{jordan2024muon} for hidden two-dimensional weight matrices and AdamW~\cite{loshchilovdecoupled} for the remaining parameters, with learning rates $10^{-3}$ and $10^{-4}$, respectively.
Both learning rates use a plateau-based schedule: if the validation loss does not improve for $20{,}000$ steps, the learning rate is decayed by a factor of $0.6$.
For the main results, we use the noise-level-dependent pocket perturbation schedule $\sigma_p=\min(0.1\sigma,0.5)$; alternative perturbation schedules are studied in Appendix Table~\ref{tab:protein_perturbation}.
All main results use a 100-step sampler with the generalized arcsin schedule described in Appendix~\ref{appendix:sampling_details}.
Across CFG scales, we use the same trained checkpoint and sampling hyperparameters, so changing the guidance scale only affects inference-time control.
We provide the code for training, sampling, and evaluation in the supplementary material.

\subsection{Main Results}
\label{subsec:main_results}
\begin{table}[htbp]
  
  \begin{center}
    
    \captionsetup{skip=4pt}
\caption{
      Summary of binding affinity, molecular properties, geometric quality, and diversity for reference molecules and molecules
      generated by \textbf{PocketVE} ($s=5$) and other baselines. $(\uparrow) / (\downarrow)$ denotes a larger / smaller number is better. Top 2
      results are highlighted with \textbf{bold text} and \uline{underlined text}, respectively.
      PocketVE is reported from one standardized run (100 pockets $\times$ 100 ligands), which is one of the three runs summarized in Appendix Table~\ref{tab:run_variability}.
      TAGMol denotes the guided variant; the unguided variant appears in Table~\ref{tab:ablation_runtime}. A dash denotes an unavailable metric.}
    \small
    \setlength{\tabcolsep}{2.5pt}
    \renewcommand{\arraystretch}{1.25}
    \label{tab1}
      \resizebox{1\columnwidth}{!}{
      \begin{tabular}{l|cc|cc|cc|cccc|cccc|c}
        \hline
        \toprule
        \multicolumn{1}{c|}{\multirow{2}{*}{Method}} & \multicolumn{2}{c|}{Vina Score $(\downarrow)$} & \multicolumn{2}{c|}{Vina Min $(\downarrow)$} & \multicolumn{2}{c|}{Vina Dock $(\downarrow)$} & \multicolumn{1}{c}{\multirow{2}{*}{\makecell{High \\ Aff. ($\uparrow$)}}} & \multicolumn{1}{c}{\multirow{2}{*}{\makecell{Joint \\ Ref. ($\uparrow$)}}} & \multicolumn{1}{c}{\multirow{2}{*}{\makecell{QED \\ Med. ($\uparrow$)}}} & \multicolumn{1}{c|}{\multirow{2}{*}{\makecell{SA \\ Med. ($\uparrow$)}}} & \multicolumn{1}{c}{\multirow{2}{*}{\makecell{Valid$_{3\text{D}}$ \\ ($\uparrow$)}}} & \multicolumn{1}{c}{\multirow{2}{*}{\makecell{Strain \\ ($\downarrow$)}}} & \multicolumn{1}{c}{\multirow{2}{*}{\makecell{Clash \\ Free ($\uparrow$)}}} & \multicolumn{1}{c|}{\multirow{2}{*}{\makecell{Centroid \\ ($\downarrow$)}}} & \multicolumn{1}{c}{\multirow{2}{*}{\makecell{Diversity \\ Med. ($\uparrow$)}}}\\

        \multicolumn{1}{c|}{}                        & Avg.                                           & Med.                                         & Avg.                                          & Med.                                  & Avg.                                 & Med.                                                                                 & \multicolumn{1}{c}{} & \multicolumn{1}{c}{} & \multicolumn{1}{c}{} & \multicolumn{1}{c|}{} & \multicolumn{1}{c}{} & \multicolumn{1}{c}{} & \multicolumn{1}{c}{} & \multicolumn{1}{c|}{} & \multicolumn{1}{c}{}      \\
        \midrule
        \multicolumn{1}{c|}{Ref}                     & -6.36                                          & -6.46                                        & -6.71                                         & -6.49                                 & -7.45                                & -7.26                                                                                & - & - & 0.47 & 0.74 &                       75&      65.3           &    97            &      -            & -         \\
        \midrule
        AR                                           & -5.75                                          & -5.64                                        & -6.18                                         & -5.88                                 & -6.75                                & -6.62                                                                                & 41.2 & 4.6 & 0.50 & 0.63          & 42.6                  & 279.1                 & 93.14                 & 1.86                  & 0.70      \\
        Pocket2Mol                                   & -5.14                                          & -4.70                                        & -6.42                                         & -5.82                                 & -7.15                                & -6.79                                                                                & 48.0 & \uline{16.5} & 0.57 & \uline{0.75} & 57.3                  & 130.6               & 81.03                 & 1.82                  & 0.71 \\
        TargetDiff                                   & -5.47                                          & -6.30                                        & -6.64                                         & -6.83                                 & -7.80                                & -7.91                                                                                & 57.6 & 5.7 & 0.48 & 0.58          & \uline{78.4}        & 306.0                 & 86.41                 & 1.49                  & 0.71 \\
        DecompDiff                                   & -5.67                                          & -6.04                                        & -7.04                                         & -7.09                                 & -8.39                                & -8.43                                                                                & 63.8 & 6.8 & 0.43 & 0.60          & 71.5                  & 318.4                 & 78.00                 & 2.90                  & 0.68      \\
        IPDiff                                       & -6.42                                          & -7.01                                        & -7.45                                         & -7.48                                 & -8.57                                & -8.51                                                                                & 68.2 & 10.0 & 0.53 & 0.59          & 52.0                  & 1248.7                & 91.27                 & 1.48                  & 0.73 \\
        TAGMol                                       & -7.02                                         & -7.77                                        & -7.95                                         & -8.07                                 & -8.59                                & -8.69                                                                                & 68.6 & 7.1 & 0.56 & 0.56          & 58.6                  & 457.4                 & 93.21                 & 1.54                  & 0.70      \\
        ALiDiff                                      & -7.07                                          & \uline{-7.95}                                & \uline{-8.09}                                 & -8.17                         & \uline{-8.90}                        & -8.81                                                                        & 68.5 & 7.0 & 0.50 & 0.56          & 54.7                  & 1239.4                & 90.85                 & \uline{1.47}          & 0.71 \\
        SeFMol                                      & -7.23                                          & -7.70                                        & -8.03                                         & -8.00                                 & -8.72                                & -8.75                                                                                & 68.7 & 11.2 & \textbf{0.64} & 0.60          & 65.2                  & 888.6                 & 93.90                 & 1.51                  & 0.68      \\
        PocketXMol                                   & -5.87                                          & -6.14                                        & -6.85                                         & -6.79                                 & -7.57                                & -7.59                                                                                & -    & -    & 0.51 & \textbf{0.80} & 67.0                  & \textbf{80.0}        & 93.46                 & 1.79                  & \textbf{0.76} \\
        BindDM                                       & -                                              & -                                            & -7.29                                         & -7.34                                 & -8.41                                & -8.37                                                                                & 64.2 & 6.9  & 0.52 & 0.58          & 62.0                  & 917.5                 & 86.60                 & 1.48                  & \uline{0.74} \\

        PAFlow                                       & \textbf{-8.31}                                 & \textbf{-8.92}                               & \textbf{-8.79}                                & \textbf{-8.96}                        & \textbf{-9.46}                       & \textbf{-9.49}                                                                       & \textbf{80.8} & 8.9 & 0.50 & 0.57          & 47.4                  & 1834.6                & \textbf{96.46}       & 1.55                  & 0.70      \\
        Ours                                         & \uline{-7.44}                                  & -7.89                                        & -7.93                                         & \uline{-8.19}                                 & -8.72                                            & \uline{-8.89}                                            & \uline{73.9} & \textbf{26.4} & \textbf{0.64} & 0.71          & \textbf{80.6}        & \uline{127.9}        & \uline{94.02}        & \textbf{1.26}        & 0.71 \\
        \hline
      \end{tabular}
      }
  \end{center}
  
\end{table}
We quantify the primary geometry comparison with a 10,000-replicate paired cluster bootstrap over complete pockets, using identical resampled pocket IDs for each method pair.
Relative to TargetDiff, PocketVE maintains comparable Valid$_{3\text{D}}$ ($+2.25$ percentage points; 95\% CI $[-0.75,5.13]$) while substantially improving geometric stability.
The corresponding strain reduction, defined as TargetDiff minus PocketVE, is $178.10$ (95\% CI $[150.38,203.44]$), making lower strain the clearest source-aligned gain.
Against guided TAGMol, the paired difference in mean Vina Score is $0.43$ in favor of PocketVE, but its 95\% CI $[-0.02,1.07]$ includes zero, so we do not claim a strict docking improvement over the architectural parent.

Figure~\ref{fig:pareto_binding_geometry} makes the binding--geometry trade-off more explicit.
Several strong-affinity baselines achieve low Vina scores but suffer from reduced Valid$_{3\text{D}}$ or high strain. In contrast, moderate CFG in PocketVE improves predicted binding while staying close to the high-validity, low-strain regime.
PAFlow illustrates this trade-off most clearly: it ranks first on the Vina-based metrics, but its much lower Valid$_{3\text{D}}$ and substantially higher strain energy suggest that the predicted affinity gain comes with a large geometric cost. The corresponding QED--validity and SA--validity trade-off plots are provided in Appendix~\ref{appendix:property_validity_tradeoff}.

\begin{figure}[tbp]

\centering
\begin{minipage}[t]{0.5\textwidth}
\centering
\textbf{Valid$_{3\text{D}}$}\par
\includegraphics[width=\linewidth]{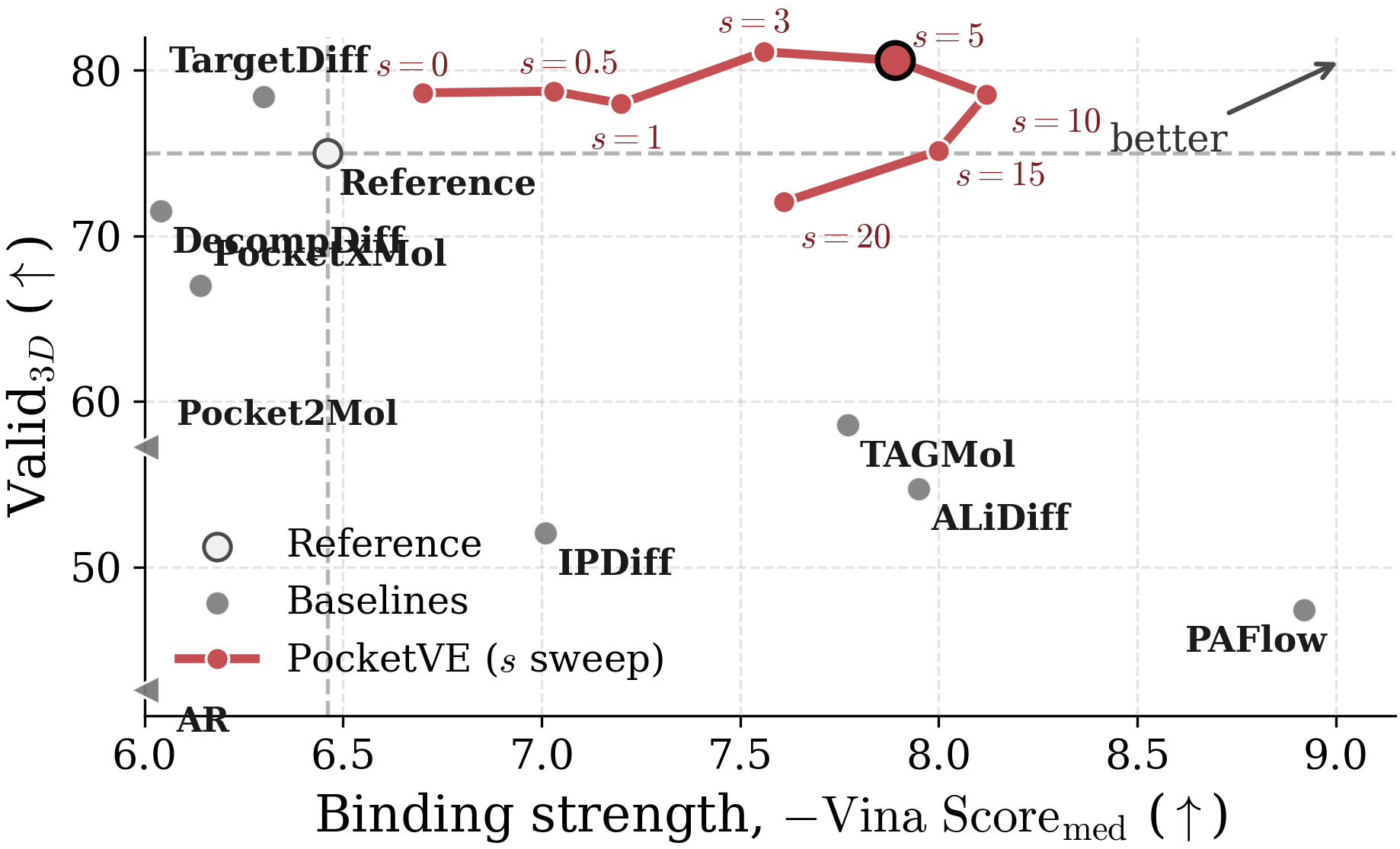}
\end{minipage}\hfill
\begin{minipage}[t]{0.5\textwidth}
\centering
\textbf{Strain Energy}\par
\includegraphics[width=\linewidth]{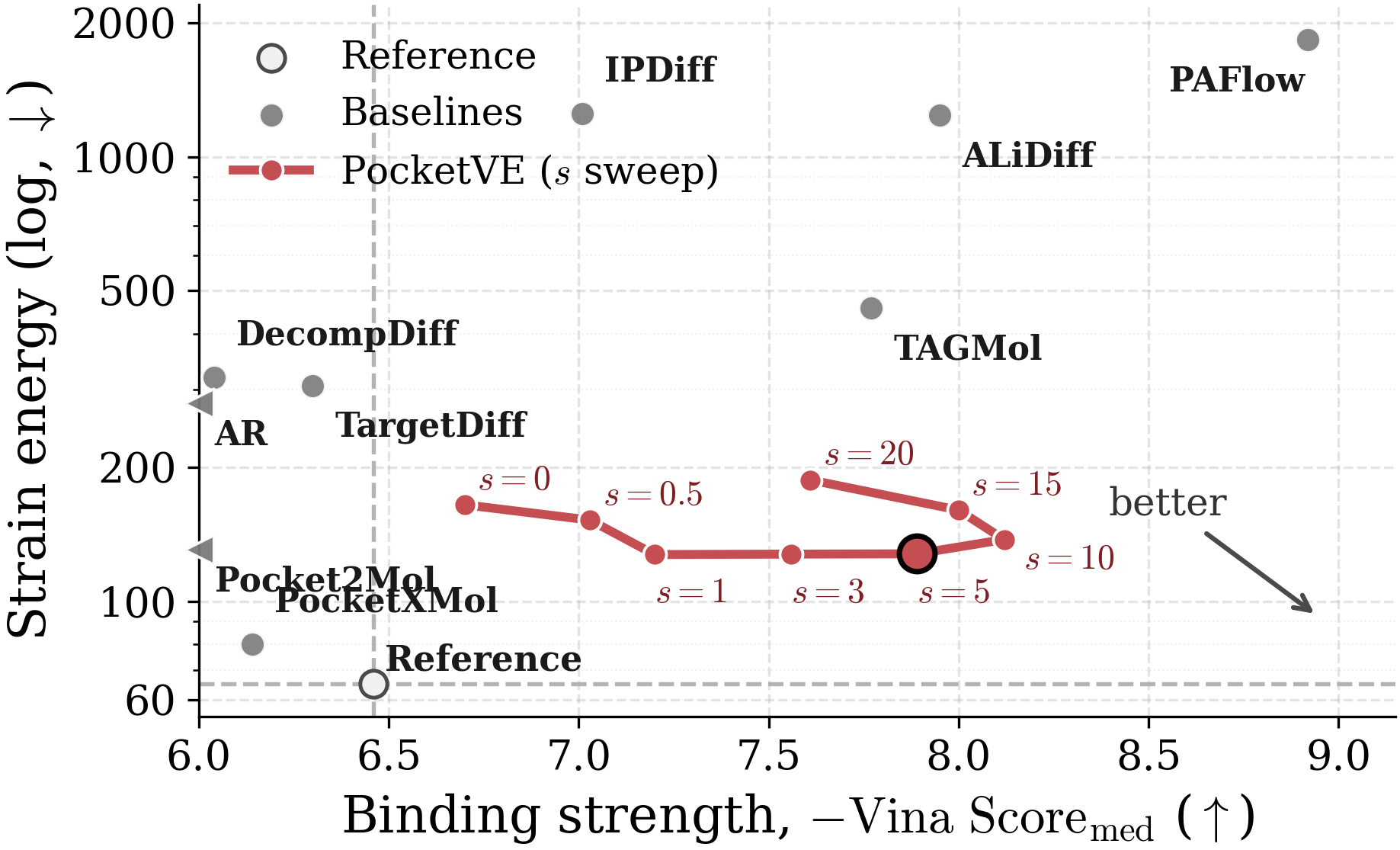}
\end{minipage}

\captionsetup{skip=4pt}
\caption{Binding--geometry trade-offs across baselines and PocketVE guidance scales.
Both panels use $-\mathrm{Vina\ Score}_{\mathrm{med}}$ on the $x$-axis, so larger values indicate stronger predicted binding.
Left: Valid$_{3\text{D}}$, where higher is better.
Right: strain energy on a log scale, where lower is better.
AR and Pocket2Mol lie outside the plotted $x$-axis range because their median Vina scores are substantially weaker ($-5.64$ and $-4.70$, respectively), and are therefore shown as off-scale markers at the left boundary.
PocketVE traces a steerable operating curve that preserves strong geometric quality under moderate CFG scales.}
\label{fig:pareto_binding_geometry}
  
\end{figure}

Beyond the aggregate GenBench3D metrics, Figure~\ref{fig:jsd_fragment_main} examines whether the geometric improvement is also visible at the distribution and fragment levels.
The JSD diagnostics show that unguided PocketVE most closely matches the reference ligand distance and atom-type distributions, while increasing the CFG scale progressively moves samples away from the reference distribution.
The rigid-fragment MMFF analysis gives a local check: compared with TargetDiff, Pocket2Mol, PAFlow, and TAGMol, PocketVE yields lower post-relaxation RMSD across most fragment sizes, especially for medium-to-large fragments, indicating more stable local geometry.
\begin{figure}[tbp]
\centering

\begin{minipage}[t]{0.54\textwidth}
\centering
\textbf{JSD Diagnostics}\par
\scriptsize
\setlength{\tabcolsep}{3.2pt}
\begin{tabular}{l|cccc}
\toprule
Method & Local $\downarrow$ & 12\AA $\downarrow$ & CC-2\AA $\downarrow$ & Atom $\downarrow$ \\
\midrule
TAGMol & 0.268 & 0.050 & 0.218 & 0.083 \\
TargetDiff & 0.234 & 0.056 & 0.213 & 0.059 \\
Pocket2Mol & 0.384 & 0.118 & 0.355 & 0.092 \\
ALiDiff & 0.341 & 0.097 & 0.350 & 0.202 \\
IPDiff & 0.391 & 0.098 & 0.349 & 0.194 \\
PAFlow & 0.454 & 0.141 & 0.461 & 0.228 \\
\midrule
Ours ($s=0$) & \textbf{0.191} & \textbf{0.035} & \textbf{0.120} & \textbf{0.047} \\
Ours ($s=5$) & 0.255 & 0.065 & 0.269 & 0.120 \\
Ours ($s=10$) & 0.284 & 0.067 & 0.292 & 0.136 \\
Ours ($s=20$) & 0.326 & 0.068 & 0.306 & 0.124 \\
\bottomrule
\end{tabular}
\end{minipage}\hfill
\begin{minipage}[t]{0.46\textwidth}
\centering
\textbf{Rigid-Fragment MMFF RMSD}\par
\includegraphics[width=\linewidth]{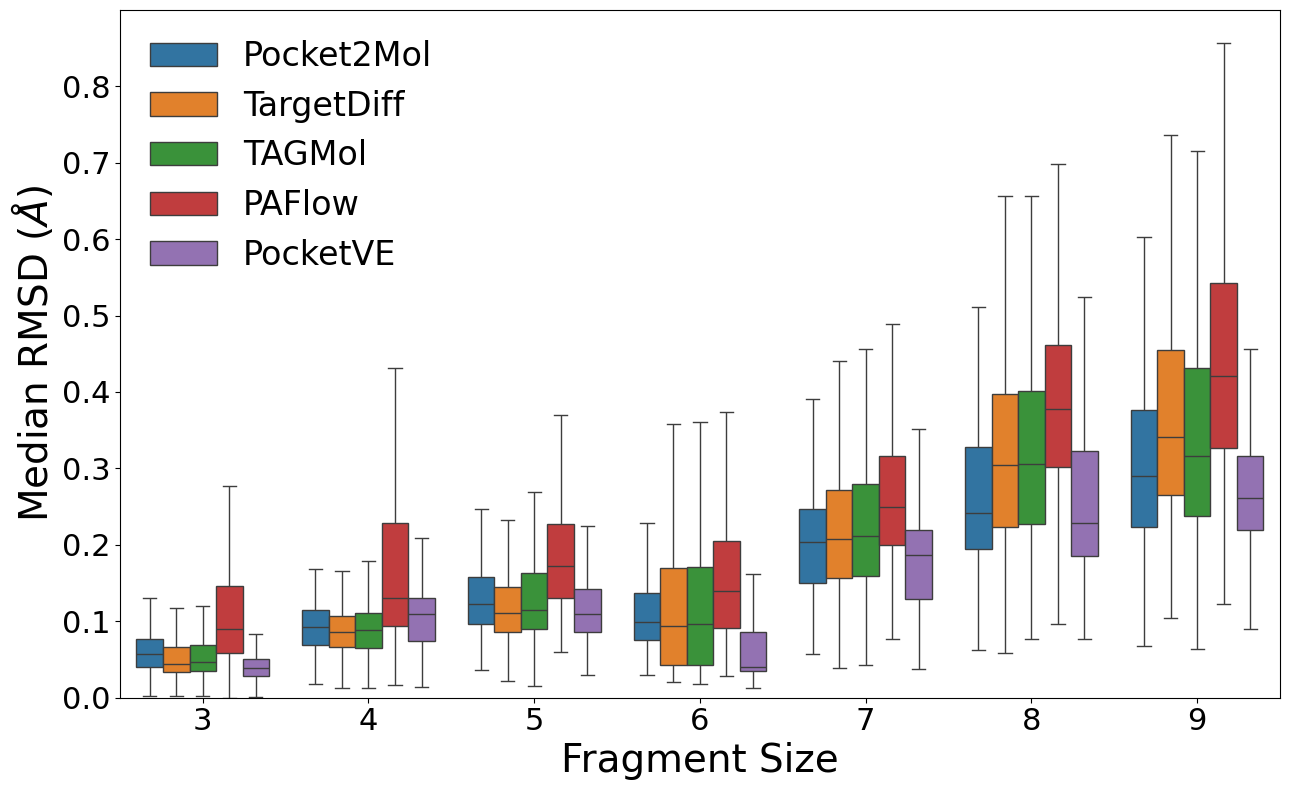}
\end{minipage}

\captionsetup{skip=4pt}
\caption{Additional geometric-fidelity analyses.
Left: JSD-based distribution diagnostics adapted from TargetDiff~\cite{guan3d}, where generated molecules are compared against reference-ligand distributions from the test set.
Lower values indicate better distribution fidelity.
Right: Median RMSD between pre- and post-MMFF relaxed rigid fragments across fragment sizes for representative baselines.}
\label{fig:jsd_fragment_main}

\end{figure}

Figure~\ref{fig:case_study_summary} provides a representative visual comparison of the same trend.
In this example, PocketVE places the ligand closer to the reference pose while maintaining a compact structure inside the pocket.
Several baseline samples show larger pose deviations or less consistent pocket occupancy.

\begin{figure}[tb]
\centering
\includegraphics[width=\textwidth]{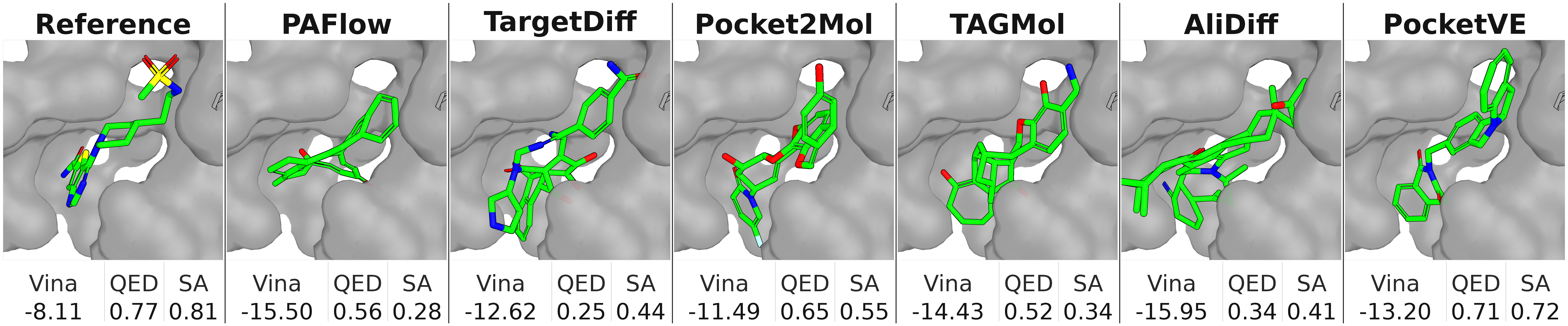}

\captionsetup{skip=4pt}
\caption{Visualizations of the reference ligand and generated ligands from baselines and PocketVE for a representative protein pocket (PDB ID: 5D7N). Vina Dock, QED, and SA are reported below.}
\label{fig:case_study_summary}

\end{figure}
In summary, PocketVE improves geometric quality while retaining competitive docking and molecular-property scores rather than uniformly dominating every metric (Table~\ref{tab1}).
Figures~\ref{fig:pareto_binding_geometry},~\ref{fig:jsd_fragment_main}, and~\ref{fig:case_study_summary} provide aggregate, distributional, fragment-level, and qualitative views of this trade-off.
Moderate CFG preserves geometric stability while shifting target-related properties, whereas stronger guidance can increase distribution drift and reduce parts of the diversity profile.
Appendix~\ref{appendix:pocket_diagnostics} reports pocket-permutation and PoseCheck diagnostics, and Appendices~\ref{appendix:run_variability} and~\ref{appendix:reference_hits} provide repeated-run and reference-normalized analyses.

\subsection{Ablations and Analysis}
\label{subsec:ablations_analysis}
We ablate the cumulative effect of the main components (Table~\ref{tab:ablation_runtime}), followed by two control axes that preserve the base architecture: the inference-time CFG scale and the training-time protein perturbation strategy.

\paragraph{Cumulative ablation}
\begin{table}[tbp]

\centering

\captionsetup{skip=4pt}
\caption{Compact ablation and efficiency summary for PocketVE.
The left table reports a cumulative ablation from TAGMol to the final PocketVE configuration: the PocketVE backbone row uses $s=0$, no protein perturbation, the log-uniform schedule, and the default sampling-time noise injection, and each subsequent row adds one component on top of the previous row.
Runtimes are measured in seconds to sample 100 ligands per protein over 100 test proteins on a single NVIDIA A100 GPU with batch size 10 and no mixed precision, excluding docking, MMFF relaxation, and other post-processing.}

\begin{minipage}[t]{0.58\textwidth}
\centering
\textbf{Cumulative Ablation}\par
\scriptsize
\setlength{\tabcolsep}{3.0pt}
\begin{tabular}{l|cccc}
\toprule
Setting & Vina Med. $\downarrow$ & Valid$_{3\text{D}}$ $\uparrow$ & Strain $\downarrow$ & Clash-Free $\uparrow$ \\
\midrule
TAGMol (unguided) & -6.08 & 65.56 & 401.3 & 78.13 \\
PocketVE backbone & -6.35 & 69.62 & 274.4 & 88.59 \\
+ gen-arcsin scheduler & -6.76 & 77.75 & 179.1 & 91.87 \\
+ CFG scale $s=1$ & -7.41 & 78.55 & 137.8 & 92.70 \\
+ CFG scale $s=5$ & -7.78 & 77.33 & \textbf{125.9} & \textbf{94.07} \\
+ protein perturb. & \textbf{-7.89} & \textbf{80.60} & 127.9 & 94.02 \\
\bottomrule
\end{tabular}
\end{minipage}\hfill
\begin{minipage}[t]{0.40\textwidth}
\centering
\textbf{Sampling Time}\par\vspace{0.3em}
\scriptsize
\setlength{\tabcolsep}{4pt}
\begin{tabular}{lc}
\toprule
Method & Time $(\downarrow)$ \\
\midrule
PocketVE ($s=1$) & \textbf{127.8} \\
PocketVE ($s=5$) & 234.1 \\
PAFlow & 249.8 \\
TargetDiff & 1382.5 \\
Pocket2Mol & 1396.8 \\
TAGMol & 2433.3 \\
\bottomrule
\end{tabular}
\end{minipage}
\label{tab:ablation_runtime}

\end{table}

The left panel of Table~\ref{tab:ablation_runtime} shows that the PocketVE backbone and sampler yield a large geometric improvement over TAGMol, while moderate CFG mainly shifts affinity-related objectives.
This comparison reflects an integrated transition in parameterization, scaling, noise, and sampling, rather than an isolated test of VE alone.
Protein perturbation then improves the final geometry-oriented operating point, raising Valid$_{3\text{D}}$ while keeping other metrics nearly unchanged.
A component-wise summary that groups the TAGMol baseline, unguided PocketVE, CFG variants, protein-perturbation variants, and schedule ablation is provided in Appendix~\ref{appendix:component_ablation}.

As a targeted coordinate-scale check, we also removed the $\sigma_{\mathrm{data}}$ factor in Eq.~\eqref{eq:appendix_scaled_input}. In the corresponding diagnostic, mean Vina changes from $-7.589$ to $+92.237$, clash-free rate collapses from $94.23\%$ to $3.39\%$, and centroid distance increases from $1.25$ to $3.33$~\AA{}, while Valid$_{3\text{D}}$ remains $78.72\%$ versus $82.39\%$. We therefore retain Eq.~\eqref{eq:appendix_scaled_input} as part of the integrated pocket-aware formulation, without attributing the full gain to this factor alone.

\paragraph{Sampling efficiency}
The PocketVE ($s=5$) runtime on the right already includes the extra CFG passes and uses the same 100-step regime as PAFlow, VEDA~\cite{zhang2026veda}, and SemlaFlow~\cite{irwin2025semlaflow}.
Under the same evaluation protocol, PocketVE remains substantially faster than most diffusion and autoregressive baselines and is slightly faster than PAFlow while remaining in the same runtime regime.
The runtime comparison on the right should be read together with the sampling-step ablation in Appendix~\ref{appendix:sampling_steps}.
In our setting, the 100-step sampler is a deliberate operating point supported by that ablation, rather than a naive truncation of a longer trajectory.

\paragraph{Guidance scale sensitivity}
We sweep the sampling-time CFG scale $s$, where $s=0$ uses the null/unconditional branch and $s=1$ recovers the plain conditional model.
As shown in Figure~\ref{fig:pareto_binding_geometry}, increasing $s$ first improves target-related objectives while largely preserving geometric quality.
However, overly large guidance scales eventually reduce Valid$_{3\text{D}}$ and increase distribution drift.
In our experiments, $s=5$ is the representative operating point, whereas $s=10$ further favors affinity-related objectives at a geometric cost.
Detailed step-wise metrics and per-metric trend plots are provided in Appendix~\ref{appendix:cfg_scale_details}.

We also evaluate property distributions, non-default requests, and conflicting conditions in Appendix~\ref{appendix:property_steering}.
These analyses show directional property shifts under CFG together with trade-offs in geometry, distributional fidelity, and diversity at stronger guidance scales.

\paragraph{Protein perturbation strategy}
To evaluate the impact of training-time protein-side perturbation, we compared fixed and noise-level-dependent schedules against a no-perturbation baseline.
Protein perturbation acts as a regularizer that reshapes the property--geometry trade-off rather than uniformly dominating the no-perturbation baseline.
Among the evaluated settings, the bounded noise-level-dependent schedule $\sigma_p=\min(0.1\sigma,0.5)$ gives the most favorable trade-off across local JSD, atom-type JSD, Vina metrics, and centroid distance.
Lighter perturbation further reduces strain, while stronger schedules improve broader distributional statistics such as JSD-All-12\AA{}, JSD-CC-2\AA{}, and SA at some cost in target-specific precision.
Detailed JSD diagnostics, full GenBench3D metrics, and an additional test-time pocket perturbation study are provided in Appendix Tables~\ref{tab:protein_perturbation} and~\ref{tab:appendix_test_time_perturbation}.

\section{Limitations}\label{sec:limitations}
Our method has several limitations that affect its practical deployment.
First, classifier-free guidance increases inference cost because each guided sampling step requires both a conditional and an unconditional forward pass.
This doubles the number of function evaluations compared with plain conditional sampling.
Second, the preferred guidance scale may be target-dependent, and stronger guidance can trade geometric and distributional fidelity for target-related objectives.
Third, the Gaussian pocket perturbation used here is a simple training regularizer rather than a realistic model of protein conformational flexibility.
Finally, although PocketVE improves geometry on the evaluated benchmark, it relies on docking-based affinity proxies and benchmark-specific protocols.
All principal results use the filtered 100-pocket CrossDocked2020 split, so broader generalization and experimental validation remain future work.
The TAGMol-to-PocketVE transition also combines several backbone and sampler choices, and the current experiments do not establish VE-only causal attribution.
Additionally, cross-method comparison is confounded by differences in generated atom count (molecular size), and all results use first-order Euler sampling, which is not claimed to be the optimal solver.
We use $\sigma_{\mathrm{data}}=10$ as a fixed operating choice; the sensitivity sweep in Appendix~\ref{appendix:coordinate_scale_diagnostics} shows that performance remains robust across the evaluated range.
Vina-based affinity proxies may reward over-optimization of binding geometry rather than true binding affinity, and the guidance scale $s=5$ was selected descriptively from the test-benchmark sweep rather than from an independent validation set.

\section{Conclusion}
\label{sec:conclusion}
In this work, we introduced PocketVE, a target-aware 3D molecular generation framework that improves geometric stability while supporting inference-time conditional control.
By revisiting coordinate parameterization, guidance design, and protein-side perturbation, PocketVE achieves higher structural quality in unguided generation and a more favorable trade-off between geometry and target objectives under guidance.
Experiments on CrossDocked2020 show improved geometric stability relative to most baselines, together with directional, aggregate multi-property steering.
Future work includes extending evaluation to broader target families, introducing matched component controls, and strengthening the connection between computational diagnostics and experimental validation.

\clearpage
\bibliographystyle{plain}
\bibliography{ref}

@article{ho2020denoising,
  title   = {Denoising diffusion probabilistic models},
  author  = {Ho, Jonathan and Jain, Ajay and Abbeel, Pieter},
  journal = {Advances in neural information processing systems},
  volume  = {33},
  pages   = {6840--6851},
  year    = {2020}
}

@inproceedings{songscore,
  title     = {Score-Based Generative Modeling through Stochastic Differential Equations},
  author    = {Song, Yang and Sohl-Dickstein, Jascha and Kingma, Diederik P and Kumar, Abhishek and Ermon, Stefano and Poole, Ben},
  booktitle = {International Conference on Learning Representations}
}

@article{karras2022elucidating,
  title   = {Elucidating the design space of diffusion-based generative models},
  author  = {Karras, Tero and Aittala, Miika and Aila, Timo and Laine, Samuli},
  journal = {Advances in neural information processing systems},
  volume  = {35},
  pages   = {26565--26577},
  year    = {2022}
}

@inproceedings{campbell2024generative,
  title     = {Generative flows on discrete state-spaces: enabling multimodal flows with applications to protein co-design},
  author    = {Campbell, Andrew and Yim, Jason and Barzilay, Regina and Rainforth, Tom and Jaakkola, Tommi},
  booktitle = {Proceedings of the 41st International Conference on Machine Learning},
  pages     = {5453--5512},
  year      = {2024}
}

@article{dhariwal2021diffusion,
  title   = {Diffusion models beat gans on image synthesis},
  author  = {Dhariwal, Prafulla and Nichol, Alexander},
  journal = {Advances in neural information processing systems},
  volume  = {34},
  pages   = {8780--8794},
  year    = {2021}
}

@article{ho2022classifier,
  title   = {Classifier-free diffusion guidance},
  author  = {Ho, Jonathan and Salimans, Tim},
  journal = {arXiv preprint arXiv:2207.12598},
  year    = {2022}
}

@inproceedings{hoogeboom2022equivariant,
  title        = {Equivariant diffusion for molecule generation in 3d},
  author       = {Hoogeboom, Emiel and Satorras, V{\i}ctor Garcia and Vignac, Cl{\'e}ment and Welling, Max},
  booktitle    = {International conference on machine learning},
  pages        = {8867--8887},
  year         = {2022},
  organization = {PMLR}
}

@inproceedings{guan3d,
  title     = {3D Equivariant Diffusion for Target-Aware Molecule Generation and Affinity Prediction},
  author    = {Guan, Jiaqi and Qian, Wesley Wei and Peng, Xingang and Su, Yufeng and Peng, Jian and Ma, Jianzhu},
  booktitle = {The Eleventh International Conference on Learning Representations}
}

@inproceedings{dorna24tagmol,
  title     = {TAGMol: Target-Aware Gradient-guided Molecule Generation},
  author    = {Dorna, Vineeth and Subhalingam, D and Kolluru, Keshav and Tuli, Shreshth and Singh, Mrityunjay and Singal, Saurabh and Krishnan, NM Anoop and Ranu, Sayan},
  booktitle = {ICML'24 Workshop ML for Life and Material Science: From Theory to Industry Applications}
}

@inproceedings{zhang2026veda,
  title     = {VEDA: Generation of 3D Molecules via Variance-Exploding Diffusion with Annealing},
  author    = {Zhang, Peining and Bi, Jinbo and Song, Minghu},
  booktitle = {Proceedings of the AAAI Conference on Artificial Intelligence},
  volume    = {40},
  number    = {33},
  pages     = {28346--28354},
  year      = {2026}
}

@article{gu2024aligning,
  title   = {Aligning target-aware molecule diffusion models with exact energy optimization},
  author  = {Gu, Siyi and Xu, Minkai and Powers, Alexander and Nie, Weili and Geffner, Tomas and Kreis, Karsten and Leskovec, Jure and Vahdat, Arash and Ermon, Stefano},
  journal = {Advances in Neural Information Processing Systems},
  volume  = {37},
  pages   = {44040--44063},
  year    = {2024}
}

@article{isert2023structure,
  title     = {Structure-based drug design with geometric deep learning},
  author    = {Isert, Clemens and Atz, Kenneth and Schneider, Gisbert},
  journal   = {Current Opinion in Structural Biology},
  volume    = {79},
  pages     = {102548},
  year      = {2023},
  publisher = {Elsevier}
}

@article{zhang2025unraveling,
  title     = {Unraveling the potential of diffusion models in small-molecule generation},
  author    = {Zhang, Peining and Baker, Daniel and Song, Minghu and Bi, Jinbo},
  journal   = {Drug Discovery Today},
  volume    = {30},
  number    = {7},
  pages     = {104413},
  year      = {2025},
  publisher = {Elsevier}
}

@article{gomez2018automatic,
  title     = {Automatic chemical design using a data-driven continuous representation of molecules},
  author    = {G{\'o}mez-Bombarelli, Rafael and Wei, Jennifer N and Duvenaud, David and Hern{\'a}ndez-Lobato, Jos{\'e} Miguel and S{\'a}nchez-Lengeling, Benjam{\'\i}n and Sheberla, Dennis and Aguilera-Iparraguirre, Jorge and Hirzel, Timothy D and Adams, Ryan P and Aspuru-Guzik, Al{\'a}n},
  journal   = {ACS central science},
  volume    = {4},
  number    = {2},
  pages     = {268--276},
  year      = {2018},
  publisher = {ACS Publications}
}

@article{loeffler2024reinvent,
  title     = {Reinvent 4: Modern AI--driven generative molecule design},
  author    = {Loeffler, Hannes H and He, Jiazhen and Tibo, Alessandro and Janet, Jon Paul and Voronov, Alexey and Mervin, Lewis H and Engkvist, Ola},
  journal   = {Journal of Cheminformatics},
  volume    = {16},
  number    = {1},
  pages     = {20},
  year      = {2024},
  publisher = {Springer}
}

@inproceedings{jin2018junction,
  title        = {Junction tree variational autoencoder for molecular graph generation},
  author       = {Jin, Wengong and Barzilay, Regina and Jaakkola, Tommi},
  booktitle    = {International conference on machine learning},
  pages        = {2323--2332},
  year         = {2018},
  organization = {PMLR}
}

@inproceedings{shigraphaf,
  title     = {GraphAF: a Flow-based Autoregressive Model for Molecular Graph Generation},
  author    = {Shi, Chence and Xu, Minkai and Zhu, Zhaocheng and Zhang, Weinan and Zhang, Ming and Tang, Jian},
  booktitle = {International Conference on Learning Representations}
}

@inproceedings{peng2022pocket2mol,
  title        = {Pocket2mol: Efficient molecular sampling based on 3d protein pockets},
  author       = {Peng, Xingang and Luo, Shitong and Guan, Jiaqi and Xie, Qi and Peng, Jian and Ma, Jianzhu},
  booktitle    = {International conference on machine learning},
  pages        = {17644--17655},
  year         = {2022},
  organization = {PMLR}
}

@article{luo20213d,
  title   = {A 3D generative model for structure-based drug design},
  author  = {Luo, Shitong and Guan, Jiaqi and Ma, Jianzhu and Peng, Jian},
  journal = {Advances in Neural Information Processing Systems},
  volume  = {34},
  pages   = {6229--6239},
  year    = {2021}
}

@article{baillif2024benchmarking,
  title   = {Benchmarking structure-based three-dimensional molecular generative models using GenBench3D: ligand conformation quality matters},
  author  = {Baillif, Benoit and Cole, Jason and McCabe, Patrick and Bender, Andreas},
  journal = {arXiv preprint arXiv:2407.04424},
  year    = {2024}
}

@article{harris2023benchmarking,
  title   = {Benchmarking generated poses: How rational is structure-based drug design with generative models?},
  author  = {Harris, Charles and Didi, Kieran and Jamasb, Arian R and Joshi, Chaitanya K and Mathis, Simon V and Lio, Pietro and Blundell, Tom},
  journal = {arXiv preprint arXiv:2308.07413},
  year    = {2023}
}

@article{francoeur2020three,
  title     = {Three-dimensional convolutional neural networks and a cross-docked data set for structure-based drug design},
  author    = {Francoeur, Paul G and Masuda, Tomohide and Sunseri, Jocelyn and Jia, Andrew and Iovanisci, Richard B and Snyder, Ian and Koes, David R},
  journal   = {Journal of chemical information and modeling},
  volume    = {60},
  number    = {9},
  pages     = {4200--4215},
  year      = {2020},
  publisher = {ACS Publications}
}

@inproceedings{huang2024protein,
  title     = {Protein-ligand interaction prior for binding-aware 3d molecule diffusion models},
  author    = {Huang, Zhilin and Yang, Ling and Zhou, Xiangxin and Zhang, Zhilong and Zhang, Wentao and Zheng, Xiawu and Chen, Jie and Wang, Yu and Cui, Bin and Yang, Wenming},
  booktitle = {The Twelfth International Conference on Learning Representations},
  year      = {2024}
}

@inproceedings{zhang2024tackling,
  title     = {Tackling the singularities at the endpoints of time intervals in diffusion models},
  author    = {Zhang, Pengze and Yin, Hubery and Li, Chen and Xie, Xiaohua},
  booktitle = {Proceedings of the IEEE/CVF Conference on Computer Vision and Pattern Recognition},
  pages     = {6945--6954},
  year      = {2024}
}

@inproceedings{lou2023reflected,
  title        = {Reflected diffusion models},
  author       = {Lou, Aaron and Ermon, Stefano},
  booktitle    = {International Conference on Machine Learning},
  pages        = {22675--22701},
  year         = {2023},
  organization = {PMLR}
}

@inproceedings{rombach2022high,
  title     = {High-resolution image synthesis with latent diffusion models},
  author    = {Rombach, Robin and Blattmann, Andreas and Lorenz, Dominik and Esser, Patrick and Ommer, Bj{\"o}rn},
  booktitle = {Proceedings of the IEEE/CVF conference on computer vision and pattern recognition},
  pages     = {10684--10695},
  year      = {2022}
}

@inproceedings{guan2023decompdiff,
  title     = {DECOMPDIFF: diffusion models with decomposed priors for structure-based drug design},
  author    = {Guan, Jiaqi and Zhou, Xiangxin and Yang, Yuwei and Bao, Yu and Peng, Jian and Ma, Jianzhu and Liu, Qiang and Wang, Liang and Gu, Quanquan},
  booktitle = {Proceedings of the 40th International Conference on Machine Learning},
  pages     = {11827--11846},
  year      = {2023}
}

@article{zhou2025prior,
  title   = {Prior-guided flow matching for target-aware molecule design with learnable atom number},
  author  = {Zhou, Jingyuan and Qian, Hao and Tu, Shikui and Xu, Lei},
  journal = {arXiv preprint arXiv:2509.01486},
  year    = {2025}
}

@article{ragoza2022generating,
  title     = {Generating 3D molecules conditional on receptor binding sites with deep generative models},
  author    = {Ragoza, Matthew and Masuda, Tomohide and Koes, David Ryan},
  journal   = {Chemical science},
  volume    = {13},
  number    = {9},
  pages     = {2701--2713},
  year      = {2022},
  publisher = {Royal Society of Chemistry}
}

@article{bickerton2012quantifying,
  title     = {Quantifying the chemical beauty of drugs},
  author    = {Bickerton, G Richard and Paolini, Gaia V and Besnard, J{\'e}r{\'e}my and Muresan, Sorel and Hopkins, Andrew L},
  journal   = {Nature chemistry},
  volume    = {4},
  number    = {2},
  pages     = {90--98},
  year      = {2012},
  publisher = {Nature Publishing Group UK London}
}

@article{ertl2009estimation,
  title     = {Estimation of synthetic accessibility score of drug-like molecules based on molecular complexity and fragment contributions},
  author    = {Ertl, Peter and Schuffenhauer, Ansgar},
  journal   = {Journal of cheminformatics},
  volume    = {1},
  number    = {1},
  pages     = {8},
  year      = {2009},
  publisher = {Springer}
}

@article{eberhardt2021autodock,
  title     = {AutoDock Vina 1.2. 0: new docking methods, expanded force field, and python bindings},
  author    = {Eberhardt, Jerome and Santos-Martins, Diogo and Tillack, Andreas F and Forli, Stefano},
  journal   = {Journal of chemical information and modeling},
  volume    = {61},
  number    = {8},
  pages     = {3891--3898},
  year      = {2021},
  publisher = {ACS Publications}
}

@article{zhou2025guiding,
  title   = {Guiding diffusion models with reinforcement learning for stable molecule generation},
  author  = {Zhou, Zhijian and An, Junyi and Liu, Zongkai and Shi, Yunfei and Zhang, Xuan and Cao, Fenglei and Qu, Chao and Qi, Yuan},
  journal = {arXiv preprint arXiv:2508.16521},
  year    = {2025}
}

@article{peng2026unified,
  title     = {Unified modeling of 3D molecular generation via atomic interactions with PocketXMol},
  author    = {Peng, Xingang and Guo, Ruihan and Guo, Fenglin and Wang, Ziyi and Sun, Jiayu and Guan, Jiaqi and Jia, Yinjun and Xu, Yan and Huang, Yanwen and Zhang, Muhan and others},
  journal   = {Cell},
  year      = {2026},
  publisher = {Elsevier}
}

@article{schneuing2024structure,
  title     = {Structure-based drug design with equivariant diffusion models},
  author    = {Schneuing, Arne and Harris, Charles and Du, Yuanqi and Didi, Kieran and Jamasb, Arian and Igashov, Ilia and Du, Weitao and Gomes, Carla and Blundell, Tom L and Lio, Pietro and others},
  journal   = {Nature Computational Science},
  volume    = {4},
  number    = {12},
  pages     = {899--909},
  year      = {2024},
  publisher = {Nature Publishing Group US New York}
}

@misc{jordan2024muon,
  author = {Keller Jordan and Yuchen Jin and Vlado Boza and Jiacheng You and
            Franz Cesista and Laker Newhouse and Jeremy Bernstein},
  title  = {Muon: An optimizer for hidden layers in neural networks},
  year   = {2024},
  url    = {https://kellerjordan.github.io/posts/muon/}
}

@inproceedings{loshchilovdecoupled,
  title     = {Decoupled Weight Decay Regularization},
  author    = {Loshchilov, Ilya and Hutter, Frank},
  booktitle = {International Conference on Learning Representations}
}

@article{jian2026general,
  title     = {General binding affinity guidance for diffusion models in structure-based drug design},
  author    = {Jian, Yue and Wu, Curtis and Reidenbach, Danny and Krishnapriyan, Aditi S},
  journal   = {Journal of Chemical Information and Modeling},
  year      = {2026},
  publisher = {ACS Publications}
}

@inproceedings{xu2023geometric,
  title        = {Geometric latent diffusion models for 3d molecule generation},
  author       = {Xu, Minkai and Powers, Alexander S and Dror, Ron O and Ermon, Stefano and Leskovec, Jure},
  booktitle    = {International Conference on Machine Learning},
  pages        = {38592--38610},
  year         = {2023},
  organization = {PMLR}
}

@inproceedings{vignac2023midi,
  title        = {Midi: Mixed graph and 3d denoising diffusion for molecule generation},
  author       = {Vignac, Clement and Osman, Nagham and Toni, Laura and Frossard, Pascal},
  booktitle    = {Joint European Conference on Machine Learning and Knowledge Discovery in Databases},
  pages        = {560--576},
  year         = {2023},
  organization = {Springer}
}

@inproceedings{irwin2025semlaflow,
  title        = {SemlaFlow--Efficient 3D Molecular Generation with Latent Attention and Equivariant Flow Matching},
  author       = {Irwin, Ross and Tibo, Alessandro and Janet, Jon Paul and Olsson, Simon},
  booktitle    = {International Conference on Artificial Intelligence and Statistics},
  pages        = {3772--3780},
  year         = {2025},
  organization = {PMLR}
}

@article{zhang2026steering,
  title     = {Steering semi-flexible molecular diffusion model for structure-based drug design with reinforcement learning},
  author    = {Zhang, Xudong and Qu, Sanqing and Lu, Fan and Wang, Jianmin and Tian, Zhixin and Gu, Shangding and Zhang, Yanping and Knoll, Alois and Gao, Shaorong and Chen, Guang and others},
  journal   = {Science Advances},
  volume    = {12},
  number    = {16},
  pages     = {eady9955},
  year      = {2026},
  publisher = {American Association for the Advancement of Science}
}

@article{bishop1995training,
  title     = {Training with noise is equivalent to Tikhonov regularization},
  author    = {Bishop, Chris M},
  journal   = {Neural computation},
  volume    = {7},
  number    = {1},
  pages     = {108--116},
  year      = {1995},
  publisher = {MIT Press}
}

@article{chapelle2000vicinal,
  title   = {Vicinal risk minimization},
  author  = {Chapelle, Olivier and Weston, Jason and Bottou, L{\'e}on and Vapnik, Vladimir},
  journal = {Advances in neural information processing systems},
  volume  = {13},
  year    = {2000}
}

@inproceedings{tishby2015deep,
  title        = {Deep learning and the information bottleneck principle},
  author       = {Tishby, Naftali and Zaslavsky, Noga},
  booktitle    = {2015 ieee information theory workshop (itw)},
  pages        = {1--5},
  year         = {2015},
  organization = {Ieee}
}

@article{nikitin2025geom,
  title     = {GEOM-drugs revisited: toward more chemically accurate benchmarks for 3D molecule generation},
  author    = {Nikitin, Filipp and Dunn, Ian and Koes, David Ryan and Isayev, Olexandr},
  journal   = {Digital Discovery},
  volume    = {4},
  number    = {11},
  pages     = {3282--3291},
  year      = {2025},
  publisher = {Royal Society of Chemistry}
}

@inproceedings{rojas2026improving,
  title     = {Improving Classifier-Free Guidance in Masked Diffusion: Low-Dim Theoretical Insights with High-Dim Impact},
  author    = {Rojas, Kevin and He, Ye and Lai, Chieh-Hsin and Takida, Yuhta and Mitsufuji, Yuki and Tao, Molei},
  booktitle = {The Fourteenth International Conference on Learning Representations},
  year      = {2026}
}

@article{groom2016cambridge,
  title     = {The Cambridge structural database},
  author    = {Groom, Colin R and Bruno, Ian J and Lightfoot, Matthew P and Ward, Suzanna C},
  journal   = {Structural Science},
  volume    = {72},
  number    = {2},
  pages     = {171--179},
  year      = {2016},
  publisher = {International Union of Crystallography}
}

@article{tong2021large,
  title     = {Large-scale analysis of bioactive ligand conformational strain energy by ab initio calculation},
  author    = {Tong, Jiahui and Zhao, Suwen},
  journal   = {Journal of Chemical Information and Modeling},
  volume    = {61},
  number    = {3},
  pages     = {1180--1192},
  year      = {2021},
  publisher = {ACS Publications}
}

@inproceedings{huang2024binding,
  title={Binding-adaptive diffusion models for structure-based drug design},
  author={Huang, Zhilin and Yang, Ling and Zhang, Zaixi and Zhou, Xiangxin and Bao, Yu and Zheng, Xiawu and Yang, Yuwei and Wang, Yu and Yang, Wenming},
  booktitle={Proceedings of the AAAI Conference on Artificial Intelligence},
  volume={38},
  number={11},
  pages={12671--12679},
  year={2024}
}
\newpage
\appendix


\section{Classifier-Free Guidance Details}
\label{appendix:cfg_details}

For completeness, we record the standard score-level form of classifier-free guidance used to motivate the implementation in the main text.

\subsection{Noise-Level Dependence of Explicit Classifier Guidance}
\label{appendix:cfg_noise_dependence}

Let $\mathcal{L}_\sigma=(\mathbf{x}_\sigma,\mathbf{z}_\sigma)$ denote the noisy ligand state at noise level $\sigma$.
The guided target distribution can be written as
\begin{equation}\label{eq:cfg_target}
  \log p_w(\mathcal{L}_\sigma \mid \mathcal{P},\mathbf{a})
  =
  (1+w)\log p(\mathcal{L}_\sigma \mid \mathcal{P},\mathbf{a})
  -w\log p(\mathcal{L}_\sigma \mid \mathcal{P},\varnothing)
  +\mathrm{const}.
\end{equation}
For the continuous coordinate branch, taking the score with respect to $\mathbf{x}_\sigma$ yields
\begin{equation}\label{eq:cfg_score}
  \nabla_{\mathbf{x}_\sigma} \log p_w(\mathcal{L}_\sigma \mid \mathcal{P},\mathbf{a})
  =
  (1+w)\nabla_{\mathbf{x}_\sigma}\log p(\mathcal{L}_\sigma \mid \mathcal{P},\mathbf{a})
  -w\nabla_{\mathbf{x}_\sigma}\log p(\mathcal{L}_\sigma \mid \mathcal{P},\varnothing).
\end{equation}
At the score level, both approaches can be read through the conditional decomposition
\begin{equation}\label{eq:cfg_bayes}
\nabla_{\mathbf{x}_\sigma}\log p(\mathcal{L}_\sigma \mid \mathcal{P},\mathbf{a})
=
\nabla_{\mathbf{x}_\sigma}\log p(\mathcal{L}_\sigma \mid \mathcal{P},\varnothing)
+
\nabla_{\mathbf{x}_\sigma}\log p(\mathbf{a}\mid \mathbf{x}_\sigma,\mathbf{z}_\sigma,\mathcal{P}),
\end{equation}
which, when substituted into Eq.~\eqref{eq:cfg_score}, gives
\begin{equation}\label{eq:cfg_score_sub}
\begin{aligned}
\nabla_{\mathbf{x}_\sigma} \log p_w(\mathcal{L}_\sigma \mid \mathcal{P},\mathbf{a})
&=
(1+w)\Big(
\nabla_{\mathbf{x}_\sigma}\log p(\mathcal{L}_\sigma \mid \mathcal{P},\varnothing)
+
\nabla_{\mathbf{x}_\sigma}\log p(\mathbf{a}\mid \mathbf{x}_\sigma,\mathbf{z}_\sigma,\mathcal{P})
\Big) \\
&\quad
-w\nabla_{\mathbf{x}_\sigma}\log p(\mathcal{L}_\sigma \mid \mathcal{P},\varnothing) \\
&=
\nabla_{\mathbf{x}_\sigma}\log p(\mathcal{L}_\sigma \mid \mathcal{P},\varnothing)
+
(1+w)\nabla_{\mathbf{x}_\sigma}\log p(\mathbf{a}\mid \mathbf{x}_\sigma,\mathbf{z}_\sigma,\mathcal{P}).
\end{aligned}
\end{equation}
This makes explicit that extrapolating the conditional score away from the dropped-condition score amplifies the same property-dependent term used by external guidance methods.
Explicit classifier guidance follows the same decomposition, but replaces that attribute-dependent term with the gradient of an external classifier or property predictor.
More concretely, it keeps the base generative score
\(
s_{\mathrm{base}}(\mathbf{x}_\sigma)
\approx
\nabla_{\mathbf{x}_\sigma}\log p(\mathcal{L}_\sigma \mid \mathcal{P},\varnothing)
\)
and approximates the attribute-dependent term by
\(
\nabla_{\mathbf{x}_\sigma}\log p_\phi(\mathbf{a}\mid \mathbf{x}_\sigma,\mathbf{z}_\sigma,\mathcal{P})
\approx
\nabla_{\mathbf{x}_\sigma}\log p(\mathbf{a}\mid \mathbf{x}_\sigma,\mathbf{z}_\sigma,\mathcal{P})
\). Introducing the same scalar coefficient $w$ as an external guidance strength then gives
\begin{equation}\label{eq:classifier_guidance_score}
s_{\mathrm{guide}}(\mathbf{x}_\sigma)
=
s_{\mathrm{base}}(\mathbf{x}_\sigma)
+
(1+w) \nabla_{\mathbf{x}_\sigma}\log p_\phi(\mathbf{a}\mid \mathbf{x}_\sigma,\mathbf{z}_\sigma,\mathcal{P}),
\end{equation}
where $s_{\mathrm{base}}$ denotes the generative score and $p_\phi$ denotes an external classifier or property predictor.
When taking the score with respect to $\mathbf{x}_\sigma$, we treat $\mathbf{z}_\sigma$ as part of the conditioning context of the joint state.
Under the VE parameterization, this relation follows from the standard Tweedie identity.
The forward process is
\(
\mathbf{x}_\sigma=\mathbf{x}_0+\sigma\boldsymbol{\epsilon},
\;
\boldsymbol{\epsilon}\sim\mathcal{N}(\mathbf{0},I)
\),
so
\(
p(\mathbf{x}_\sigma\mid\mathbf{x}_0)=\mathcal{N}(\mathbf{x}_0,\sigma^2I)
\).
For fixed conditioning context $(\mathbf{z}_\sigma,\mathcal{P},\mathbf{a})$, Tweedie's formula gives
\begin{equation}
\mathbb{E}[\mathbf{x}_0\mid \mathbf{x}_\sigma,\mathbf{z}_\sigma,\mathcal{P},\mathbf{a}]
=
\mathbf{x}_\sigma+\sigma^2\nabla_{\mathbf{x}_\sigma}\log p(\mathbf{x}_\sigma\mid \mathbf{z}_\sigma,\mathcal{P},\mathbf{a}).
\end{equation}
Approximating this posterior mean with the denoiser prediction
\(
D_\theta(\mathbf{x}_\sigma,\mathbf{z}_\sigma,\sigma,\mathcal{P},\mathbf{a})
\approx
\mathbb{E}[\mathbf{x}_0\mid \mathbf{x}_\sigma,\mathbf{z}_\sigma,\mathcal{P},\mathbf{a}]
\)
and rearranging yields
\begin{equation}\label{eq:cfg_x0_score}
\nabla_{\mathbf{x}_\sigma}\log p(\mathbf{x}_\sigma \mid \mathbf{z}_\sigma,\mathcal{P},\mathbf{a})
\approx
\frac{D_\theta(\mathbf{x}_\sigma,\mathbf{z}_\sigma,\sigma,\mathcal{P},\mathbf{a})-\mathbf{x}_\sigma}{\sigma^2}.
\end{equation}
Applying the same VE identity to the explicit classifier-guidance score in Eq.~\eqref{eq:classifier_guidance_score} yields the equivalent $x_0$-space correction
\begin{equation}\label{eq:classifier_guidance_x0}
D_{\mathrm{guide}}(\mathbf{x}_\sigma,\mathbf{z}_\sigma,\sigma,\mathcal{P})
\approx
D_{\mathrm{base}}(\mathbf{x}_\sigma,\mathbf{z}_\sigma,\sigma,\mathcal{P})
+
w \sigma^2 \nabla_{\mathbf{x}_\sigma}\log p_\phi(\mathbf{a}\mid \mathbf{x}_\sigma,\mathbf{z}_\sigma,\mathcal{P}).
\end{equation}
Applying Eq.~\eqref{eq:cfg_x0_score} to both the conditional and dropped-condition branches in Eq.~\eqref{eq:cfg_score} gives
\begin{equation}
\nabla_{\mathbf{x}_\sigma}\log p_w
\approx
\frac{(1+w)D_\theta(\mathbf{x}_\sigma,\mathbf{z}_\sigma,\sigma,\mathcal{P},\mathbf{a})
-wD_\theta(\mathbf{x}_\sigma,\mathbf{z}_\sigma,\sigma,\mathcal{P},\varnothing)
-\mathbf{x}_\sigma}{\sigma^2}.
\end{equation}
Under the probability-flow ODE parameterization of VE models~\cite{songscore,karras2022elucidating}, the reverse update field induced by a clean prediction $\hat{\mathbf{x}}_0$ can be written as
\begin{equation}
\mathbf{d}
=
\frac{\mathbf{x}_\sigma-\hat{\mathbf{x}}_0}{\sigma}.
\end{equation}
Applying this identity to the base predictor
\(
D_{\mathrm{base}}(\mathbf{x}_\sigma,\mathbf{z}_\sigma,\sigma,\mathcal{P})
\)
gives the unguided field
\begin{equation}
\mathbf{d}_{\mathrm{base}}
=
\frac{\mathbf{x}_\sigma-D_{\mathrm{base}}(\mathbf{x}_\sigma,\mathbf{z}_\sigma,\sigma,\mathcal{P})}{\sigma}.
\end{equation}
Substituting the classifier-guided clean prediction from Eq.~\eqref{eq:classifier_guidance_x0} then recovers the main-text update rule
\begin{equation}
\mathbf{d}_{\mathrm{guide}}
=
\mathbf{d}_{\mathrm{base}}
- w \sigma \nabla_{\mathbf{x}_\sigma}\log p_\phi(\mathbf{a}\mid \mathbf{x}_\sigma,\mathbf{z}_\sigma,\mathcal{P}),
\end{equation}
which is Eq.~\eqref{eq:classifier_guidance_ode} in the main text.
This makes the noise-level dependence explicit: an external classifier gradient enters the denoiser through a $\sigma^2$ factor, and the sampler update through a $\sigma$ factor.

\subsection{Recovering the Classifier-Free Guidance Rule}
\label{appendix:cfg_recover}

Returning to Eq.~\eqref{eq:cfg_score}, we now recover the main-text implementation of classifier-free guidance.
Comparing the expression above with the VE identity
\begin{equation}
\nabla_{\mathbf{x}_\sigma}\log p_w
=
\frac{\hat{\mathbf{x}}_\theta^{\mathrm{cfg}}-\mathbf{x}_\sigma}{\sigma^2},
\end{equation}
we identify the guided clean prediction as
\begin{equation}
\hat{\mathbf{x}}_\theta^{\mathrm{cfg}}
=
(1+w)D_\theta(\mathbf{x}_\sigma,\mathbf{z}_\sigma,\sigma,\mathcal{P},\mathbf{a})
-wD_\theta(\mathbf{x}_\sigma,\mathbf{z}_\sigma,\sigma,\mathcal{P},\varnothing).
\end{equation}
Using the implementation parameterization $s=w+1$ (equivalently, $w=s-1$) then yields
\begin{equation}
\hat{\mathbf{x}}_\theta^{\mathrm{cfg}}
=
D_\theta(\mathbf{x}_\sigma,\mathbf{z}_\sigma,\sigma,\mathcal{P},\varnothing)
+ s\!\left(
D_\theta(\mathbf{x}_\sigma,\mathbf{z}_\sigma,\sigma,\mathcal{P},\mathbf{a})
- D_\theta(\mathbf{x}_\sigma,\mathbf{z}_\sigma,\sigma,\mathcal{P},\varnothing)\right),
\end{equation}
which is exactly the coordinate branch in Eq.~\eqref{eq:cfg_prediction} of the main text.
For the discrete branch, we do not use a continuous score over $\mathbf{z}_\sigma$.
Instead, we apply the same conditional / dropped-condition extrapolation directly to the categorical logits:
\begin{equation}
\hat{H}_\theta^{\mathrm{cfg}}
=
H_\theta(\mathbf{x}_\sigma,\mathbf{z}_\sigma,\sigma,\mathcal{P},\varnothing)
+ s\!\left(
H_\theta(\mathbf{x}_\sigma,\mathbf{z}_\sigma,\sigma,\mathcal{P},\mathbf{a})
- H_\theta(\mathbf{x}_\sigma,\mathbf{z}_\sigma,\sigma,\mathcal{P},\varnothing)\right).
\end{equation}
Applying $\mathrm{softmax}(\hat{H}_\theta^{\mathrm{cfg}})$ then recovers the discrete branch in Eq.~\eqref{eq:cfg_prediction}.

\section{EDM Preconditioning and Target-Aware Coordinate Scaling}
\label{appendix:edm_preconditioning}

This section summarizes the EDM preconditioning recipe used by PocketVE and distinguishes the standard VE/EDM machinery from the target-aware adaptation introduced by our pocket-conditioned setting.
The coefficient definitions themselves are standard EDM/Karras-style components reused from VEDA~\cite{karras2022elucidating,zhang2026veda}; the main PocketVE-specific point is how the ligand coordinates are scaled before entering the shared protein--ligand geometric graph.

Under the VE corruption process, clean ligand coordinates are perturbed as
\begin{equation}\label{eq:appendix_ve_corruption}
\mathbf{x}_\sigma = \mathbf{x}_0 + \sigma \boldsymbol{\epsilon},
\qquad
\boldsymbol{\epsilon}\sim\mathcal{N}(\mathbf{0},I).
\end{equation}
Standard EDM writes the coordinate denoiser in the preconditioned form
\begin{equation}\label{eq:appendix_edm_precond}
D_\theta(\mathbf{x}_\sigma,\sigma)
=
c_{\mathrm{skip}}(\sigma)\mathbf{x}_\sigma
+
c_{\mathrm{out}}(\sigma)\,
F_\theta\!\left(c_{\mathrm{in}}(\sigma)\mathbf{x}_\sigma,\,
c_{\mathrm{noise}}(\sigma)\right),
\end{equation}
with
\begin{equation}\label{eq:appendix_edm_coeffs_a}
c_{\mathrm{skip}}(\sigma)=\frac{\sigma_{\mathrm{data}}^2}{\sigma^2+\sigma_{\mathrm{data}}^2},
\qquad
c_{\mathrm{out}}(\sigma)=\frac{\sigma\,\sigma_{\mathrm{data}}}{\sqrt{\sigma^2+\sigma_{\mathrm{data}}^2}},
\end{equation}
\begin{equation}\label{eq:appendix_edm_coeffs_b}
c_{\mathrm{in}}(\sigma)=\frac{1}{\sqrt{\sigma^2+\sigma_{\mathrm{data}}^2}},
\qquad
c_{\mathrm{noise}}(\sigma)=\frac{\log \sigma}{4}.
\end{equation}
Here $c_{\mathrm{in}}$ normalizes the noisy input magnitude, $c_{\mathrm{skip}}$ keeps the prediction anchored to the noisy sample at low noise, $c_{\mathrm{out}}$ rescales the network residual back to the clean-coordinate domain, and $c_{\mathrm{noise}}$ provides the scalar noise embedding.

PocketVE keeps the standard output preconditioning and noise embedding, while modifying how the EDM input scaling is applied before the pocket--ligand graph.
At the notation level used in the main text, this is written as
\begin{equation}\label{eq:appendix_pocketve_precond}
D_\theta(\mathbf{x}_\sigma,\mathbf{z}_\sigma,\sigma,\mathcal{P})
=
c_{\mathrm{skip}}(\sigma)\mathbf{x}_\sigma
+
c_{\mathrm{out}}(\sigma)\,
F_\theta\!\left(\tilde{\mathbf{x}}_\sigma;\, c_{\mathrm{noise}}(\sigma),\, \mathbf{z}_\sigma,\, \mathcal{P}\right).
\end{equation}
In the code, this is realized through an equivalent two-stage implementation: the network first operates on the scaled ligand coordinates
\begin{equation}\label{eq:appendix_scaled_input}
\mathbf{x}_{\mathrm{net}}
=
\sigma_{\mathrm{data}}\,c_{\mathrm{in}}(\sigma)\mathbf{x}_\sigma
=
\tilde{\mathbf{x}}_\sigma,
\end{equation}
and predicts the coordinates in that scaled frame as
\begin{equation}\label{eq:appendix_scaled_output}
\hat{\mathbf{x}}_{\mathrm{net}}
=
F_\theta\!\left(\mathbf{x}_{\mathrm{net}};\, c_{\mathrm{noise}}(\sigma),\, \mathbf{z}_\sigma,\, \mathcal{P}\right),
\end{equation}
and the final coordinate prediction in the default constant mode is reconstructed as
\begin{equation}\label{eq:appendix_output_reconstruct}
\hat{\mathbf{x}}_0
=
c_{\mathrm{skip}}(\sigma)\mathbf{x}_\sigma
+
c_{\mathrm{out}}(\sigma)\big(\hat{\mathbf{x}}_{\mathrm{net}}-\mathbf{x}_{\mathrm{net}}\big).
\end{equation}
This makes explicit that the network prediction is interpreted relative to the scaled ligand input before being mapped back to the original coordinate frame.
The discrete atom-type branch is handled separately through the DFM-based categorical predictor, so the present discussion is only about the continuous coordinate preconditioning.

The target-aware adaptation appears in the network input.
In textbook EDM, one would feed $c_{\mathrm{in}}(\sigma)\mathbf{x}_\sigma$ into the coordinate backbone.
In the actual PocketVE implementation, however, the ligand coordinates are effectively fed as
\begin{equation}\label{eq:appendix_scale_preserve_input}
\sigma_{\mathrm{data}}\,c_{\mathrm{in}}(\sigma)\mathbf{x}_\sigma
=
\frac{\sigma_{\mathrm{data}}}{\sqrt{\sigma^2+\sigma_{\mathrm{data}}^2}}\mathbf{x}_\sigma
=
\tilde{\mathbf{x}}_\sigma.
\end{equation}
The same scaled input is used as the reference point in the output reconstruction, so the model predicts a correction relative to the scaled ligand coordinates rather than absolute coordinates in the network frame.
The design is necessary because ligand and protein coordinates interact inside the same distance-based equivariant graph, while the pocket coordinates are not scaled by EDM coefficients.
If the ligand alone were shrunk by the textbook $c_{\mathrm{in}}(\sigma)$ factor, then at low noise the ligand would be rescaled by approximately $1/\sigma_{\mathrm{data}}$ while the pocket remained in the original coordinate frame, distorting protein--ligand distances.
The scale-preserving input together with the matching output recentering avoids that mismatch.
As $\sigma\to 0$, it satisfies $\tilde{\mathbf{x}}_\sigma\to \mathbf{x}_\sigma$, so the ligand remains in the same physical frame as the pocket near the clean limit, while at larger noise levels it still attenuates the coordinate magnitude in the same spirit as EDM.

This target-aware scaling also explains why PocketVE uses a comparatively large $\sigma_{\mathrm{data}}$.
Our coordinates are normalized around the pocket center rather than around each ligand itself, so $\sigma_{\mathrm{data}}$ must cover both the internal ligand extent and the ligand displacement from the pocket center.
For the same reason, we keep $\sigma_{\min}=10^{-3}$ small enough to preserve near-clean coordinate precision in the final reverse steps.
For training, the continuous coordinate branch uses the EDM-style weighting induced by the implementation loss
\(
\lVert \hat{\mathbf{x}}_0-\mathbf{x}_0\rVert_2^2 / c_{\mathrm{out}}(\sigma)^2
\),
while the categorical branch uses a separately weighted cross-entropy term.
The concrete dataset statistics behind these choices, together with the full schedule specification, are given in Appendix~\ref{appendix:sampling_details}.

\section{Sampling Details}
\label{appendix:sampling_details}

Training was run on a single NVIDIA A100 GPU for approximately 24 hours.
The host machine used four AMD EPYC 7513 32-Core Processor CPUs and 100~GB of memory.

For the VE coordinate branch, we set $\sigma_{\max}=800$, $\sigma_{\min}=10^{-3}$, and $\sigma_{\mathrm{data}}=10.0$.
The choice of $\sigma_{\mathrm{data}}=10.0$ is target-aware rather than ligand-only.
On the training set, the per-ligand coordinate spread has mean and median standard deviation 3.82~\AA{}, with a maximum of 9.91~\AA{}.
Because our coordinates are centered by the protein pocket rather than by each ligand itself, the relevant clean-data scale must also cover ligand displacement from the pocket center.
Under this normalization, the ligand-center distance has mean 2.31~\AA{}, median 1.98~\AA{}, and maximum 9.64~\AA{}.
We therefore choose $\sigma_{\mathrm{data}}=10.0$ so that the preconditioning scale covers both the internal ligand extent and the pocket-relative ligand offset.
We also keep $\sigma_{\min}=10^{-3}$ small enough to preserve near-clean coordinate precision at the end of the reverse trajectory.

We traverse the noise levels with a generalized arcsin schedule, which extends the arcsin scheduler used in VEDA~\cite{zhang2026veda}.
Let $\tau_i=1-i/N$ and
\begin{equation}\label{eq:gen_arcsin_schedule}
\sigma_i
=
\exp\!\left(
\log\sigma_{\min}
+
\left[
(1-\beta)\tau_i
+
\beta\frac{2}{\pi}\arcsin(\tau_i^p)
\right]
(\log\sigma_{\max}-\log\sigma_{\min})
\right),
\end{equation}
where we use $\beta=2.2$ and $p=0.57$ in all experiments.
Let $\sigma_0>\sigma_1>\cdots>\sigma_N$ denote the resulting noise schedule.

For the discrete atom types, we first convert the type head into a clean categorical prediction
\begin{equation}\label{eq:sampling_type_pred}
\mathbf{p}_{\theta,i}
=
\mathrm{softmax}\!\big(H_\theta(\mathbf{x}_{\sigma_i},\mathbf{z}_{\sigma_i},\sigma_i,\mathcal{P})\big).
\end{equation}
We then update $\mathbf{z}_{\sigma_i}$ with a discrete sampler based on the DFM transition-rate formulation~\cite{campbell2024generative}.
In this view, the transition rate toward category $j$ is
\begin{equation}\label{eq:sampling_dfm_rate}
R_\theta(z_{\sigma_i},j)
=
\omega(\sigma_i)\,p_\theta(z_0=j\mid z_{\sigma_i})
+
\eta_i\,p_\theta(z_0=z_{\sigma_i}\mid z_{\sigma_i}),
\end{equation}
where
\begin{equation}\label{eq:sampling_dfm_omega}
\omega(\sigma_i)
=
\frac{\eta_i S(1-m(\sigma_i))+\eta_i m(\sigma_i)+m'(\sigma_i)}{m(\sigma_i)},
\end{equation}
$S$ is the number of atom categories, $p_\theta(z_0=\cdot\mid z_{\sigma_i})$ is given by the predicted categorical distribution from Eq.~\eqref{eq:sampling_type_pred}, and we follow the variable-noise setting in \cite{campbell2024generative} with $\eta_i=\eta/m'(\sigma_i)$.
Equivalently, this step can be interpreted as masked-token refinement under the same schedule $m(\sigma)$.

\section{Sampling-Step Ablation}
\label{appendix:sampling_steps}

We ablate the number of reverse sampling steps using representative CFG scales $s\in\{0,1,5,10\}$.
Figure~\ref{fig:appendix_sampling_steps} shows that moving from 25 to 100 steps improves both predicted binding strength and 3D validity across guidance scales, while increasing from 100 to 200 steps gives smaller additional gains relative to the extra sampling cost.
We therefore use $N=100$ as the default operating point in the main experiments.

\begin{figure}[htbp]
\centering
\begin{minipage}[t]{0.49\textwidth}
\centering
\includegraphics[width=\linewidth]{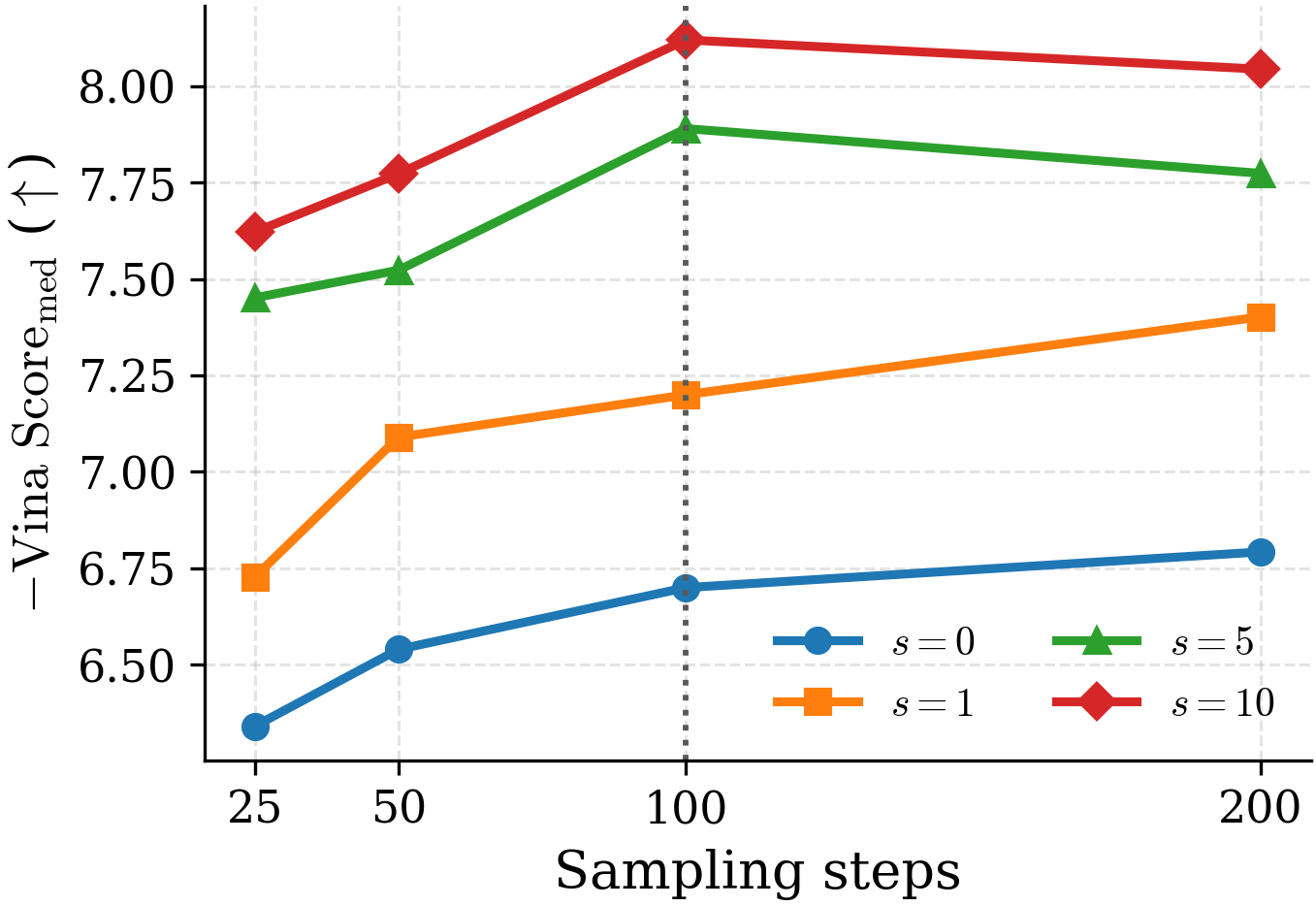}
\end{minipage}\hfill
\begin{minipage}[t]{0.49\textwidth}
\centering
\includegraphics[width=\linewidth]{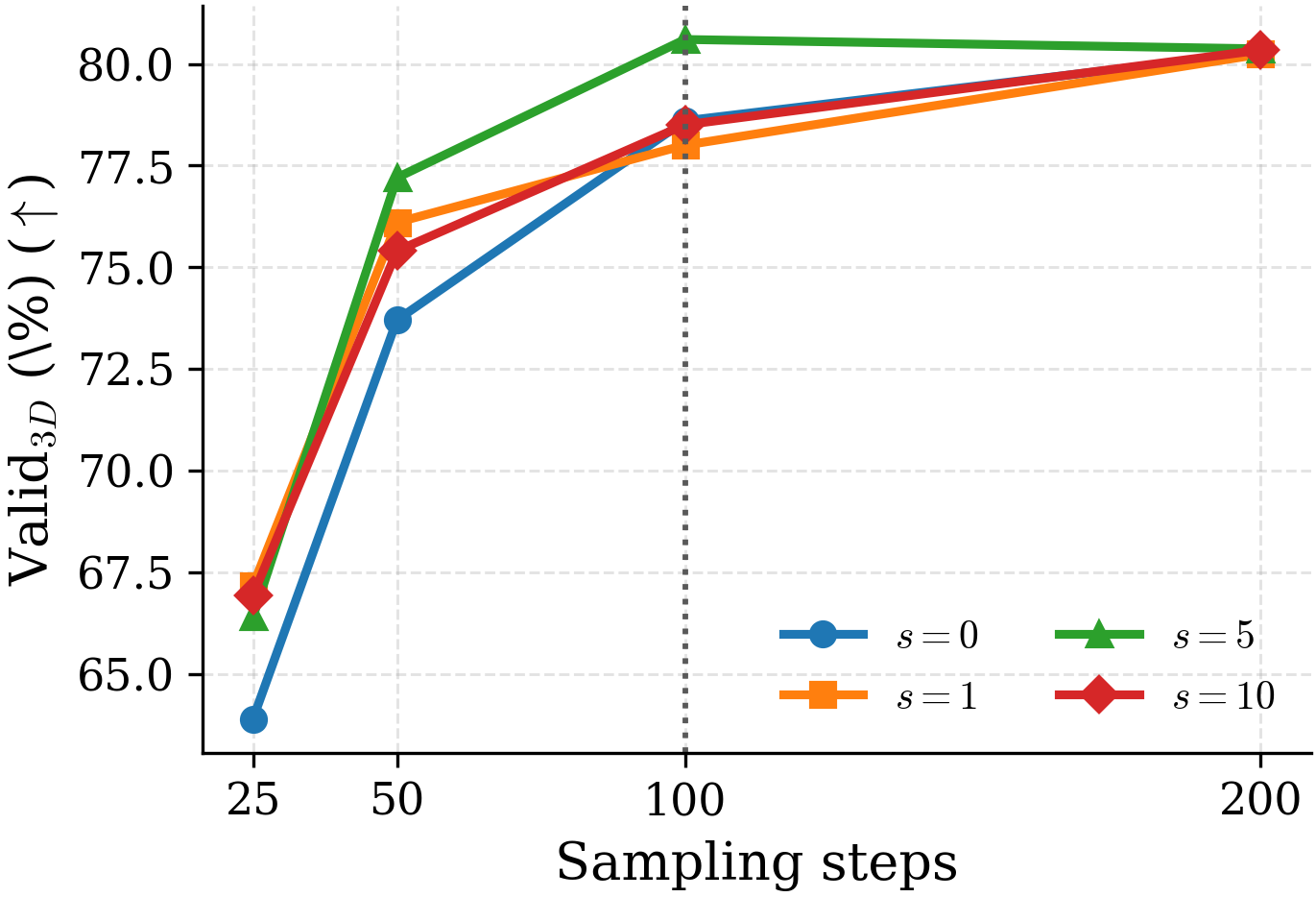}
\end{minipage}

\captionsetup{skip=4pt}
\caption{Sampling-step ablation under representative CFG scales.
Left: median Vina Score is shown as $-\mathrm{Vina\ Score}_{\mathrm{med}}$, so higher is better.
Right: Valid$_{3\text{D}}$ under the GenBench3D protocol.
The vertical dashed line marks the default setting $N=100$.}
\label{fig:appendix_sampling_steps}
\end{figure}

\begin{table}[htbp]
\centering

\captionsetup{skip=4pt}
\caption{Sampling-step ablation for representative CFG scales.
We report median Vina Score, Valid$_{3\text{D}}$, and median strain energy.
Lower Vina Score and strain energy are better; higher Valid$_{3\text{D}}$ is better.}
\label{tab:appendix_sampling_steps}
\scriptsize
\setlength{\tabcolsep}{5pt}
\begin{tabular}{ccccc}
\toprule
$s$ & Steps & Vina Score $(\downarrow)$ & Valid$_{3\text{D}}$ $(\uparrow)$ & Strain $(\downarrow)$ \\
\midrule
0 & 25 & -6.34 & 63.89 & 398.3 \\
0 & 50 & -6.54 & 73.70 & 231.3 \\
0 & 100 & -6.70 & 78.60 & 164.9 \\
0 & 200 & -6.79 & 80.29 & 144.4 \\
1 & 25 & -6.73 & 67.17 & 324.1 \\
1 & 50 & -7.09 & 76.10 & 178.9 \\
1 & 100 & -7.20 & 78.00 & 127.4 \\
1 & 200 & -7.40 & 80.24 & 116.5 \\
5 & 25 & -7.45 & 66.40 & 293.0 \\
5 & 50 & -7.52 & 77.20 & 158.8 \\
5 & 100 & -7.89 & \textbf{80.60} & 127.9 \\
5 & 200 & -7.77 & 80.37 & \textbf{111.2} \\
10 & 25 & -7.62 & 66.93 & 326.5 \\
10 & 50 & -7.77 & 75.40 & 172.7 \\
10 & 100 & \textbf{-8.12} & 78.50 & 137.7 \\
10 & 200 & -8.04 & 80.34 & 118.9 \\
\bottomrule
\end{tabular}
\end{table}

\section{Pocket Perturbation Details}
\label{appendix:pocket_perturbation_details}

In our implementation, pocket perturbation is applied directly to pocket coordinates:
\begin{equation}\label{eq:pocket_jitter}
  \tilde{\mathbf{y}} = \mathbf{y} + \boldsymbol{\eta}, \qquad \boldsymbol{\eta}\sim \mathcal{N}(\mathbf{0}, \sigma_p^2 \mathbf{I}),
\end{equation}
while pocket atom features are kept unchanged.

The following local statement formalizes the regularization view of pocket perturbation using the standard connection between input noise, vicinal risk minimization, and Jacobian regularization~\cite{bishop1995training,chapelle2000vicinal}.

\begin{proposition}[Pocket perturbation regularizes conditional sensitivity]
Let $f_\theta(\cdot,\mathcal{P})$ denote either the coordinate denoising head or the pre-softmax type head at a fixed ligand noise level $\sigma$, and let $\ell(f_\theta(\cdot,\mathcal{P}),t)$ be a smooth per-sample denoising loss with target $t$.
For Gaussian pocket perturbation $\boldsymbol{\eta}\sim\mathcal{N}(0,\sigma_p^2 I)$, the perturbed-pocket objective satisfies
\begin{equation}\label{eq:pocket_taylor}
  \mathbb{E}_{\boldsymbol{\eta}}\,
  \ell(f_\theta(\cdot,\mathcal{P}+\boldsymbol{\eta}),t)
  =
  \ell(f_\theta(\cdot,\mathcal{P}),t)
  + \frac{\sigma_p^2}{2}
  \mathrm{Tr}\!\left(\nabla_{\mathcal{P}}^2
  \ell(f_\theta(\cdot,\mathcal{P}),t)\right)
  + O(\sigma_p^4).
\end{equation}
under standard smoothness assumptions.
For a squared coordinate loss, the leading nonnegative term contains
$\frac{\sigma_p^2}{2}\lVert J_{\mathcal{P}} f_\theta\rVert_F^2$ near the optimum.
Thus pocket perturbation penalizes sharp dependence of the denoiser on small pocket-coordinate changes.
Equivalently, the model is trained against a local neighborhood of mildly perturbed pocket views rather than a single rigid conditioning instance.
When the perturbation scale is small, this neighborhood can be interpreted as approximating nearby pocket-coordinate perturbations.
\end{proposition}

\begin{proof}
Apply a second-order Taylor expansion to $\ell(f_\theta(\cdot,\mathcal{P}+\boldsymbol{\eta}),t)$ around $\mathcal{P}$.
The first-order term vanishes after expectation because $\mathbb{E}[\boldsymbol{\eta}]=0$, and the second-order term contracts the pocket Hessian with $\mathbb{E}[\boldsymbol{\eta}\boldsymbol{\eta}^{\top}]=\sigma_p^2 I$, giving Eq.~\eqref{eq:pocket_taylor}.
For squared denoising loss, the Hessian decomposes into a Gauss--Newton term $J_{\mathcal{P}} f_\theta^\top J_{\mathcal{P}} f_\theta$ plus residual-weighted second-derivative terms; near a well-fit denoiser the residual terms are small, leaving the stated Jacobian penalty.
\end{proof}

This local analysis also explains why the perturbation can be made noise-level dependent.
At large ligand noise $\sigma$, the denoiser should use coarse pocket information and can tolerate stronger pocket smoothing; at small $\sigma$, the reverse process refines local geometry and should not blur steric constraints.
We therefore use a bounded schedule $\sigma_p=\min(0.1\sigma,0.5)$ in the main setting: the proportional term keeps pocket uncertainty below the ligand corruption scale, while the cap prevents the augmentation from erasing Angstrom-scale pocket geometry.
From an information-bottleneck perspective~\cite{tishby2015deep}, this schedule suppresses nuisance information about the exact crystallographic pocket realization while preserving the pocket information needed to predict ligand geometry and affinity.
If $\sigma_p$ is too small, the Jacobian penalty is negligible and the model can overfit to brittle pocket details; if it is too large, the conditioning channel loses target-specific information and the model shifts toward a weaker, less precise conditional distribution.
In this sense, pocket perturbation controls the local Lipschitz behavior of the pocket-conditioning channel: it smooths the denoiser's response to pocket geometry and only indirectly smooths the effective conditional landscape explored by the sampler.
The ablation in Table~\ref{tab:protein_perturbation} is consistent with this bias--variance trade-off: the mild bounded schedule improves local distance fidelity, atom-type JS, Vina metrics, and centroid distance, the corresponding unbounded $0.1\sigma$ schedule remains competitive but is slightly weaker on the same target-specific metrics, the lighter $\min(0.05\sigma,0.25)$ schedule further reduces strain, and the stronger $\min(0.3\sigma,1.5)$ schedule favors broader distributional and SA-oriented statistics but gives up some target-specific precision.

\begin{table}[htbp]
\centering

\captionsetup{skip=4pt}
\caption{Sensitivity to protein-side perturbation strategies. Local JSD averages the eight short-range pairwise distance JSDs over $6$-$6|1,2,4$, $6$-$7|1,2,4$, and $6$-$8|1,2$, corresponding to C--C, C--N, and C--O distance statistics.
JSD-All-12\AA{} pools all intraligand atom-pair distances up to 12~\AA{}, JSD-CC-2\AA{} focuses on short-range carbon--carbon distances, and Atom reports atom-type frequency JS.
Lower JSD, Vina metrics, strain energy, and centroid distance are better; higher QED, SA, Valid$_{3\text{D}}$, and clash-free rate are better.}
\label{tab:protein_perturbation}
\scriptsize
\setlength{\tabcolsep}{2.8pt}
\resizebox{\linewidth}{!}{
\begin{tabular}{l|cccc|cccc|cccc}
\toprule
 & \multicolumn{4}{c|}{JSD $(\downarrow)$} & \multicolumn{4}{c|}{Property / Affinity} & \multicolumn{4}{c}{Geometry} \\
\cmidrule(lr){2-5}\cmidrule(lr){6-9}\cmidrule(lr){10-13}
Perturbation & Local & 12\AA{} & CC-2\AA{} & Atom & QED $(\uparrow)$ & SA $(\uparrow)$ & \makecell{Vina\\$(\downarrow)$} & \makecell{Vina\\Min $(\downarrow)$} & \makecell{Valid$_{3\text{D}}$\\$(\uparrow)$} & \makecell{Strain\\$(\downarrow)$} & \makecell{Clash-\\Free $(\uparrow)$} & \makecell{Centroid\\$(\downarrow)$} \\
\midrule
None & 0.270 & 0.069 & 0.276 & 0.143 & 0.63 & 0.73 & -7.78 & -7.99 & 77.3 & 125.9 & 94.07 & 1.271 \\
Fixed $0.1$ & 0.262 & 0.064 & 0.274 & 0.140 & \textbf{0.65} & 0.72 & -7.70 & -8.15 & 80.0 & 129.2 & 93.15 & 1.257 \\
$0.1\sigma$ & 0.267 & 0.068 & 0.269 & 0.140 & 0.64 & 0.74 & -7.83 & -8.07 & \textbf{81.3} & 122.5 & 94.27 & 1.307 \\
$\min(0.05\sigma,0.25)$ & 0.270 & 0.070 & 0.288 & 0.145 & 0.64 & 0.74 & -7.86 & -8.16 & 81.2 & \textbf{121.0} & 94.13 & 1.275 \\
$\min(0.1\sigma,0.5)$ & \textbf{0.255} & 0.065 & 0.269 & \textbf{0.120} & 0.64 & 0.71 & \textbf{-7.89} & \textbf{-8.19} & 80.6 & 127.9 & 94.02 & \textbf{1.255} \\
$\min(0.2\sigma,1.0)$ & 0.269 & 0.069 & 0.281 & 0.137 & 0.64 & 0.74 & -7.77 & -8.08 & 80.2 & 121.8 & \textbf{94.28} & 1.286 \\
$\min(0.3\sigma,1.5)$ & 0.257 & \textbf{0.063} & \textbf{0.263} & 0.132 & 0.64 & \textbf{0.75} & -7.66 & -7.98 & 79.4 & 118.5 & 93.67 & 1.280 \\
\bottomrule
\end{tabular}
}

\end{table}

\section{Guidance Scale Sensitivity Details}
\label{appendix:cfg_scale_details}

Table~\ref{tab:cfg_scale} and Figures~\ref{fig:cfg_scale_property}--\ref{fig:cfg_scale_geometry} provide the step-wise CFG-scale diagnostics corresponding to the main-text trade-off trajectory in Figure~\ref{fig:pareto_binding_geometry}.
They make the same trend explicit from metric-wise views: weak-to-moderate guidance improves affinity-related and molecular-property objectives while preserving strong geometry, whereas overly strong guidance eventually reduces Valid$_{3\text{D}}$ and increases strain and centroid error.

\begin{table}[htbp]
\centering

\captionsetup{skip=4pt}
\caption{Sensitivity to CFG scale. We report median property metrics together with key GenBench3D geometry metrics.}
\label{tab:cfg_scale}
\scriptsize
\setlength{\tabcolsep}{3.5pt}
\begin{tabular}{c|ccc|cccc}
\toprule
$s$ & QED $(\uparrow)$ & SA $(\uparrow)$ & Vina Score $(\downarrow)$ & Valid$_{3\text{D}}$ $(\uparrow)$ & Strain $(\downarrow)$ & Clash-Free $(\uparrow)$ & Centroid Dist. $(\downarrow)$ \\
\midrule
0 & 0.48 & 0.63 & -6.70 & 78.64 & 164.9 & 90.59 & 1.25 \\
0.5 & 0.57 & 0.66 & -7.03 & 78.74 & 152.3 & 92.02 & \textbf{1.24} \\
1 & 0.60 & 0.70 & -7.20 & 77.99 & 127.4 & 90.74 & 1.25 \\
2 & 0.62 & 0.71 & -7.33 & 78.92 & \textbf{125.6} & 92.15 & 1.26 \\
3 & 0.63 & 0.71 & -7.56 & \textbf{81.12} & 127.7 & 93.57 & 1.25 \\
5 & \textbf{0.64} & 0.71 & -7.89 & 80.60 & 127.9 & 94.02 & 1.26 \\
10 & \textbf{0.64} & \textbf{0.72} & \textbf{-8.12} & 78.54 & 137.7 & \textbf{95.03} & 1.28 \\
15 & \textbf{0.64} & 0.70 & -8.00 & 75.14 & 160.4 & 94.89 & 1.31 \\
20 & 0.60 & 0.69 & -7.61 & 72.04 & 187.5 & 94.51 & 1.38 \\
\bottomrule
\end{tabular}

\end{table}

\begin{figure}[htbp]
\begin{minipage}{0.48\textwidth}
\centering
\includegraphics[width=\linewidth]{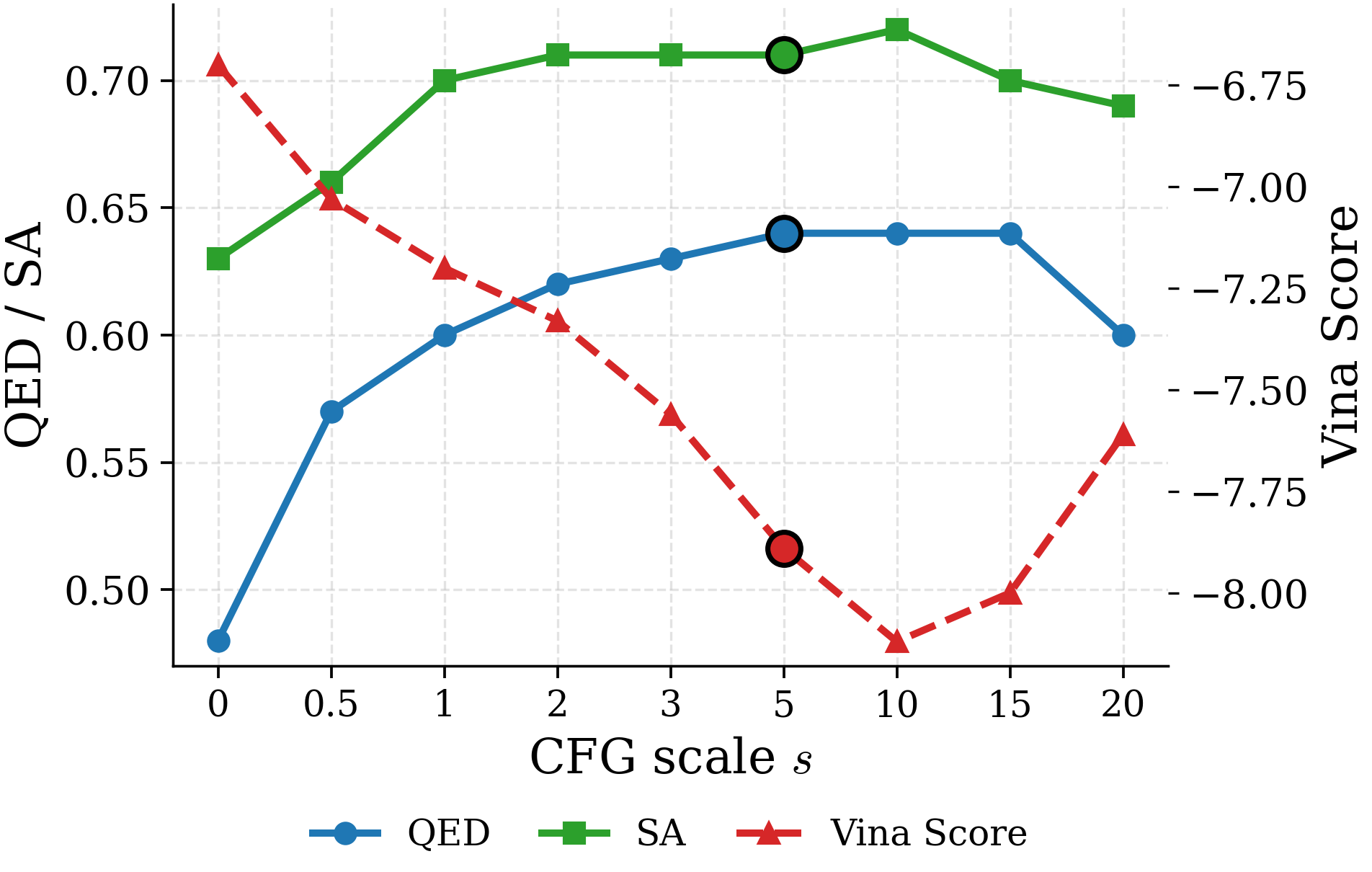}

\captionsetup{skip=4pt}
\caption{Property-side trends under different CFG scales. QED, SA, and Vina Score improve from weak to moderate guidance, with Vina Score reaching its strongest value around $s=10$ and QED/SA saturating around the moderate-to-strong range.}
\label{fig:cfg_scale_property}
\end{minipage}\hfill
\begin{minipage}{0.48\textwidth}
\centering
\includegraphics[width=\linewidth]{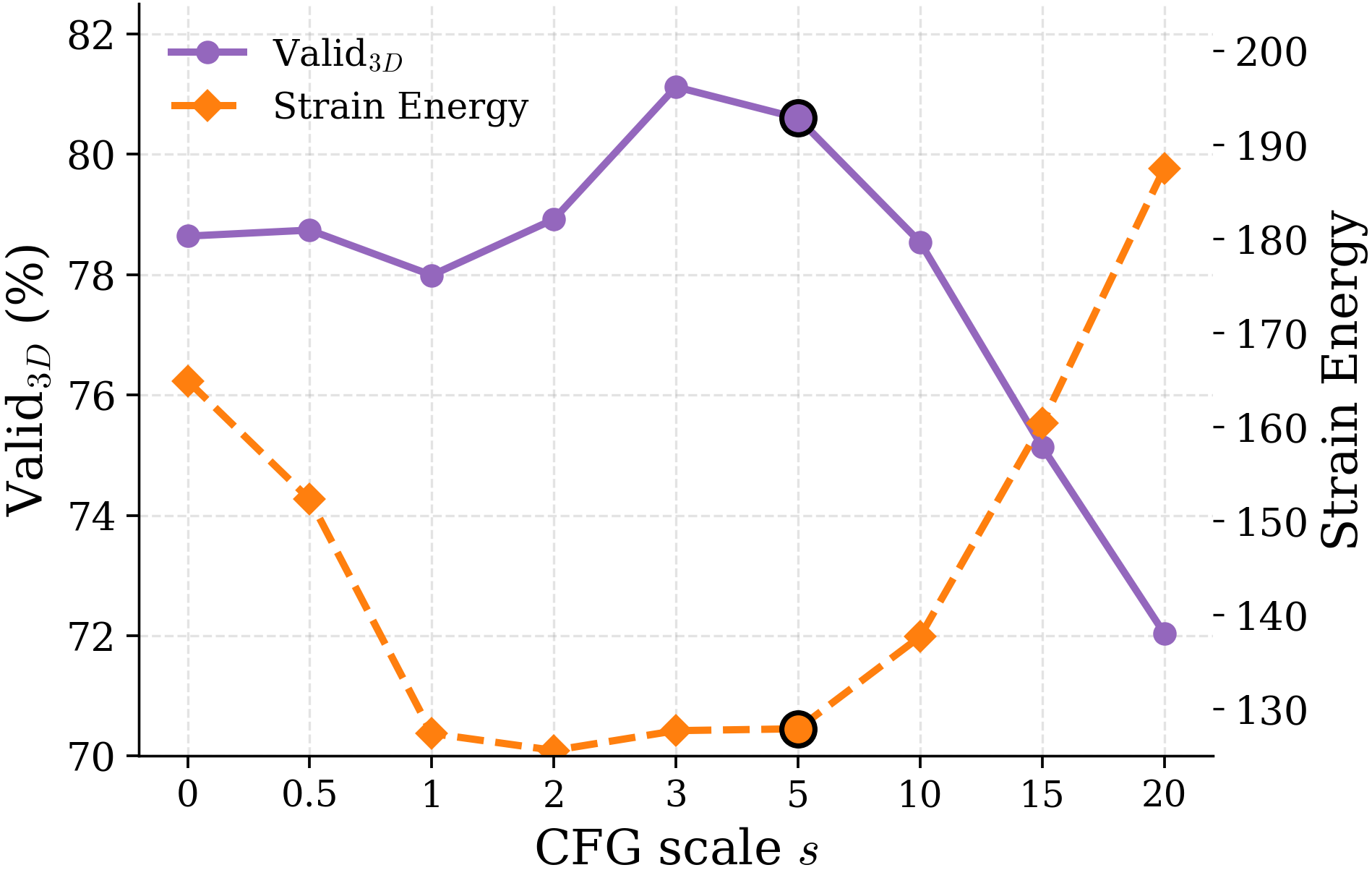}

\captionsetup{skip=4pt}
\caption{Geometry-side trends under different CFG scales. Moderate guidance keeps Valid$_{3\text{D}}$ high and strain energy low, whereas overly strong guidance reduces Valid$_{3\text{D}}$ and increases strain. In this study, $s=5$ is used as the representative operating point.}
\label{fig:cfg_scale_geometry}
\end{minipage}

\end{figure}

\section{RampUp Guidance Schedule}
\label{appendix:rampup_guidance}

In addition to the constant CFG scale used in the main experiments, we also consider a simple RampUp variant inspired by recent analyses of guidance schedules in masked diffusion~\cite{rojas2026improving}.
The motivation is that strong extrapolation at the beginning of sampling can be brittle, since the ligand state is still dominated by high noise and the conditional and dropped-condition predictions may be poorly calibrated.
Instead of applying a fixed scale $s$ at every reverse step, RampUp uses a time-dependent scale
\begin{equation}
  s_i^{\mathrm{ramp}} = \alpha_i s,\qquad
  \alpha_i = \frac{i}{N-1}, \qquad i=0,\ldots,N-1,
\end{equation}
where $N$ is the number of sampling steps and $i$ indexes the reverse trajectory from the highest-noise step to the final low-noise step.
The guided coordinate and atom-type predictions then use the same CFG rule as Eq.~\eqref{eq:cfg_prediction}, with $s$ replaced by $s_i^{\mathrm{ramp}}$ at step $i$.
Thus, early steps remain close to the property-unconditional ($s=0$) trajectory, while later low-noise steps receive the full property-steering strength.
This schedule does not require retraining or additional model evaluations; it only changes the scalar guidance multiplier used during inference.

We use RampUp as an auxiliary diagnostic rather than as the default sampler.
Since the same nominal scale is not directly comparable across constant and RampUp schedules, we compare them at similar median docking strength.
In the mid-guidance regime, RampUp gives a comparable trade-off: for example, constant CFG with $s=2$ gives median Vina Score $-7.33$ and median Vina Min $-7.78$, while RampUp with final scale $s=5$ gives $-7.49$ and $-7.81$.
At this matched operating point, RampUp keeps a similar distributional profile (JSD-All $0.0488$ versus $0.0490$) and slightly higher completion rate ($0.960$ versus $0.954$), with comparable molecular stability ($34.7\%$ versus $34.4\%$).
At stronger docking targets, however, this advantage is not consistent; matching the constant $s=5$ median Vina Score with RampUp requires a larger final scale and lowers molecular stability.
We therefore treat RampUp as a conservative schedule that can improve the median trade-off in a moderate guidance range, rather than as a uniformly better replacement for constant CFG.

\section{Component-Wise Ablation Summary}
\label{appendix:component_ablation}

To make the contribution of each design choice easier to inspect, we collect the main component-wise comparisons in Table~\ref{tab:component_ablation}.
The table reorganizes results by design axis rather than by metric category.
Several entries are repeated from Tables~\ref{tab1},~\ref{tab:cfg_scale}, and~\ref{tab:protein_perturbation} to make the attribution easier to read in one place.
All PocketVE variants use the same test split and GenBench3D evaluation protocol, and guided variants use the same trained checkpoint unless the ablated component changes training.
The schedule row reports an additional ablation, while the other rows reorganize results from the main tables by design axis.
\begin{table}[htbp]
\centering

\captionsetup{skip=4pt}
\caption{Component-wise ablation summary for PocketVE.}
\label{tab:component_ablation}
\scriptsize
\setlength{\tabcolsep}{3.2pt}
\begin{tabular}{l|l|ccc|cccc}
\toprule
Ablation axis & Setting & QED & SA & Vina Score & Valid$_{3\text{D}}$ & Strain & Clash-Free & Centroid \\
\midrule
Backbone/guidance baseline & TAGMol (guided) & 0.56 & 0.56 & -7.77 & 58.6 & 457.4 & 93.21 & 1.54 \\
VE backbone without CFG & PocketVE ($s=0$) & 0.48 & 0.63 & -6.70 & 78.64 & 164.9 & 90.59 & 1.25 \\
Plain conditional sampling & PocketVE ($s=1$) & 0.60 & 0.70 & -7.20 & 77.99 & 127.4 & 90.74 & 1.25 \\
Moderate CFG & PocketVE ($s=5$) & 0.64 & 0.71 & -7.89 & 80.60 & 127.9 & 94.02 & 1.26 \\
Protein perturbation & No perturbation & 0.63 & 0.73 & -7.78 & 77.3 & 125.9 & 94.07 & 1.271 \\
Optimizer & Adam ($s=5$) & 0.64 & 0.73 & -7.84 & 78.34 & 152.8 & 93.74 & 1.272 \\
Optimizer & AdamW ($s=5$) & 0.63 & 0.73 & -7.40 & 76.62 & 170.5 & 93.58 & 1.340 \\
Noise schedule & Log-uniform schedule ($s=5$) & 0.64 & 0.70 & -7.13 & 73.11 & 186.8 & 90.77 & 1.303 \\
Sampler control & No sampling-time noise & 0.59 & 0.63 & -5.75 & 59.70 & 495.7 & 75.70 & 1.44 \\
\bottomrule
\end{tabular}

\end{table}

The comparison between TAGMol (guided) and PocketVE ($s=0$) should be read as the bundled effect of switching to the PocketVE backbone and sampler without guidance, including the VE-style coordinate parameterization, sampling-time noise injection, and generalized arcsin schedule. Because the TAGMol reference itself uses guidance, this gap also includes the removal of TAGMol's own guidance mechanism.
Together with the cumulative ablation in Table~\ref{tab:ablation_runtime}, this shows that the geometric stabilization mainly comes from these backbone-level sampling and parameterization changes before inference-time guidance is applied.
Concretely, Valid$_{3\text{D}}$ increases by about 20 points, while strain energy is reduced by about 64\% relative to guided TAGMol.
The comparison from $s=0$ to $s=5$ then isolates inference-time CFG under the same checkpoint and sampling procedure: moderate guidance improves Vina Score, QED, and SA while preserving high 3D validity.
The no-perturbation row shows that pocket-side perturbation is not a uniformly dominant switch, but it changes the property--geometry frontier and provides a useful reference point in the current setting.
The optimizer rows provide controls at the same CFG scale.
Adam, which follows the optimizer choice used in TAGMol, remains a strong baseline, while AdamW-only training is weaker in both affinity and geometry.
The main Muon--AdamW setting still gives the best overall geometry and strain, suggesting that the final gains are not simply due to replacing Adam with AdamW.
The log-uniform-schedule row further shows that the reported gains are not produced by CFG alone: replacing the generalized arcsin schedule with log-uniform sampling at the same guidance scale leads to weaker Vina Score, lower Valid$_{3\text{D}}$, higher strain, and worse centroid error.

\section{Repeated-Run Variability}
\label{appendix:run_variability}

To document the stability of the final operating point without forcing the reader to cross-reference the main table, Table~\ref{tab:run_variability} reports the final PocketVE metrics together with their repeated-run standard deviations under the same evaluation protocol.
We keep this summary compact and focus on the main distribution, property, affinity, and geometry metrics used throughout the paper.

\begin{table}[htbp]
\centering

\captionsetup{skip=4pt}
\caption{Final PocketVE performance with repeated-run variability.
Each entry is the median across three independent runs $\pm$ standard deviation; the Table~\ref{tab1} run is one of the three.
We report the same representative statistics used in the main text: QED$_{\mathrm{med}}$, SA$_{\mathrm{med}}$, Vina Score$_{\mathrm{med}}$, Vina Min$_{\mathrm{med}}$, Valid$_{3\text{D}}$, Strain$_{\mathrm{med}}$, Clash-Free, Centroid mean, Local, JSD-All-12\AA{}, JSD-CC-2\AA{}, and atom-type JS.
Valid$_{3\text{D}}$ and Clash-Free use the same displayed units as the main table.}
\label{tab:run_variability}
\scriptsize
\setlength{\tabcolsep}{3.2pt}
\begin{tabular}{cccccccc}
\toprule
QED $\uparrow$ & SA $\uparrow$ & Vina $\downarrow$ & Vina Min $\downarrow$ & Valid$_{3\text{D}}$ $\uparrow$ & Strain $\downarrow$ & Clash-Free $\uparrow$ & Centroid $\downarrow$ \\

\midrule
0.63 $\pm$ 0.0018 & 0.72 $\pm$ 0.0047 & -7.734 $\pm$ 0.0150 & -8.026 $\pm$ 0.0204 & 79.27 $\pm$ 0.68 & 124.7 $\pm$ 1.29 & 93.89 $\pm$ 5.91 & 1.284 $\pm$ 0.0127 \\
\midrule
Local $\downarrow$ & 12\AA{} $\downarrow$ & CC-2\AA{} $\downarrow$ & Atom $\downarrow$ & \multicolumn{4}{c}{} \\
\midrule
0.245 $\pm$ 0.0020 & 0.0629 $\pm$ 0.0003 & 0.265 $\pm$ 0.0004 & 0.121 $\pm$ 0.0003 & \multicolumn{4}{c}{} \\

\bottomrule
\end{tabular}

\end{table}

\section{Reference-Normalized Multi-Objective Hits}
\label{appendix:reference_hits}

Table~\ref{tab:reference_hits} reports an auxiliary reference-normalized diagnostic rather than a replacement for the main benchmark metrics.
The goal is to test whether a method produces same-pocket samples that are jointly no worse than the reference ligand in docking affinity, QED, and SA.
Among the filtered methods, PocketVE achieves the highest Joint-Ref average and Joint Success rate.

\begin{table}[htbp]
\centering

\captionsetup{skip=4pt}
\caption{Reference-normalized multi-objective hit metrics; all percentage metrics are higher-is-better.
Overall QED, SA, and diversity averages are included to preserve the average property statistics omitted from the compact main table.
High Affinity follows prior SBDD usage and denotes the per-pocket fraction of generated molecules whose Vina Dock score is no worse than the reference ligand, and Success is the percentage of pockets with at least one such molecule.
High-Aff. Subset reports average QED and SA within the High Affinity subset.
Hit Rate is the fixed-threshold multi-objective rate with QED $\geq 0.4$, SA $\geq 0.5$, and Vina Dock $\leq -8.18$ kcal/mol.
Joint-Ref is the per-pocket fraction of generated molecules that simultaneously satisfy Vina Dock, QED, and SA no worse than the reference ligand.
Joint Success is the percentage of pockets with at least one Joint-Ref molecule.}
\label{tab:reference_hits}
\scriptsize
\setlength{\tabcolsep}{3.2pt}
\begin{tabular}{l|c|ccc|ccc|cc|ccc}
\toprule
\multirow{2}{*}{Method} & \multirow{2}{*}{Hit Rate} & \multicolumn{3}{c|}{Overall Avg.} & \multicolumn{3}{c|}{High Affinity} & \multicolumn{2}{c|}{High-Aff. Subset} & \multicolumn{3}{c}{Joint-Ref} \\
 & & QED & SA & Div. & Avg. & Med. & Success & QED Avg. & SA Avg. & Avg. & Med. & Success \\
\midrule
AR & 12.9 & 0.51 & 0.63 & 0.70 & 41.2 & 37.0 & 72.2 & 0.52 & 0.59 & 4.6 & 0.0 & 36.1 \\
Pocket2Mol & 24.3 & 0.56 & \textbf{0.74} & 0.69 & 48.0 & 51.0 & 88.0 & 0.57 & \textbf{0.72} & \uline{16.5} & \uline{8.0} & \uline{74.0} \\
TargetDiff & 20.5 & 0.48 & 0.58 & 0.72 & 57.6 & 58.3 & 99.0 & 0.50 & 0.56 & 5.7 & 2.0 & 58.0 \\
DecompDiff & 27.7 & 0.45 & 0.61 & 0.68 & 63.8 & 75.0 & 91.9 & 0.44 & 0.59 & 6.8 & 1.1 & 50.5 \\
IPDiff & 27.4 & 0.52 & 0.61 & \textbf{0.74} & 68.2 & 74.6 & 98.0 & 0.52 & 0.57 & 10.0 & 2.6 & 70.0 \\
TAGMol (guided) & 27.7 & 0.55 & 0.56 & 0.69 & 68.6 & 75.7 & 99.0 & 0.55 & 0.54 & 7.1 & 1.5 & 58.0 \\
BindDM & 24.9 & 0.51 & 0.58 & - & 64.2 & 66.1 & \textbf{100.0} & 0.53 & 0.55 & 6.9 & 1.2 & 67.0 \\
ALiDiff & 25.2 & 0.50 & 0.57 & \uline{0.73} & 68.5 & 76.0 & \textbf{100.0} & 0.52 & 0.54 & 7.0 & 0.0 & 48.0 \\
SeFMol & \uline{37.4} & \textbf{0.63} & 0.60 & 0.69 & 68.7 & 76.3 & \textbf{100.0} & \textbf{0.63} & 0.57 & 11.2 & 4.0 & 68.0 \\
PAFlow & 31.5 & 0.49 & 0.57 & 0.71 & \textbf{80.8} & \textbf{93.7} & 99.0 & 0.50 & 0.56 & 8.9 & 1.1 & 52.0 \\
\midrule
PocketVE & \textbf{47.3} & \uline{0.61} & \uline{0.73} & 0.72 & \uline{73.9} & \uline{85.2} & \textbf{100.0} & \uline{0.62} & \uline{0.70} & \textbf{26.4} & \textbf{20.9} & \textbf{92.0} \\
\bottomrule
\end{tabular}

\end{table}

\section{Additional Property--Validity Trade-Offs}
\label{appendix:property_validity_tradeoff}

Figure~\ref{fig:appendix_pareto_property_valid3d} extends the main binding--geometry analysis to the drug-likeness and synthesizability axes.
The same qualitative trade-off remains visible: weak-to-moderate guidance improves the property side without immediately leaving the high-validity regime, whereas stronger guidance eventually pushes the model toward lower Valid$_{3\text{D}}$.
We place these plots in the appendix because they reinforce the same guidance-scale story as Figure~\ref{fig:pareto_binding_geometry}, but from auxiliary property views rather than the primary affinity view.

\begin{figure}[htbp]
\centering
\includegraphics[width=0.95\columnwidth]{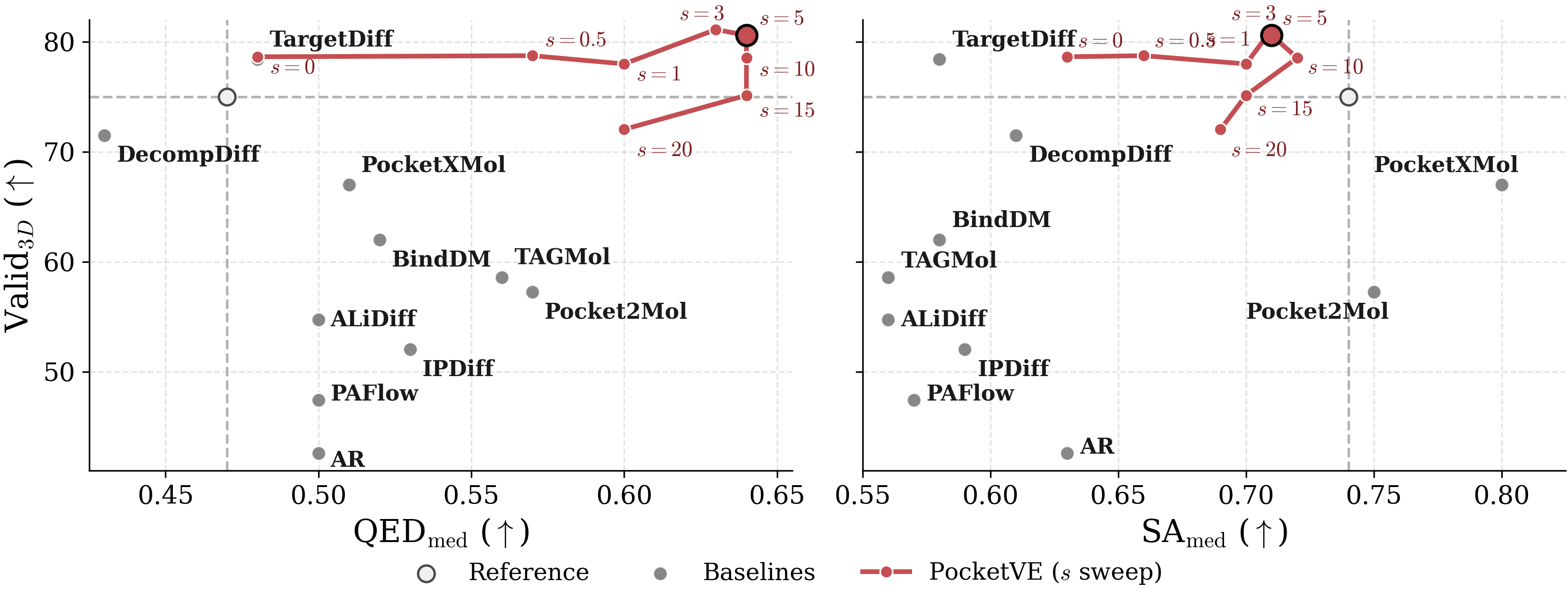}

\captionsetup{skip=4pt}
\caption{Property--validity trade-off across baselines and PocketVE guidance scales.
The left and right panels report QED--Valid$_{3\text{D}}$ and SA--Valid$_{3\text{D}}$, respectively.
Only baselines with available GenBench3D Valid$_{3\text{D}}$ values are shown.
PocketVE moves toward higher drug-likeness and synthesizability while retaining high 3D validity under moderate CFG scales.}
\label{fig:appendix_pareto_property_valid3d}

\end{figure}

\section{Additional Distance-Distribution Diagnostics}
\label{appendix:jsd_overlay}

Figure~\ref{fig:appendix_jsd_overlay} makes the JSD diagnostics in the main text visually inspectable rather than purely numeric.
Its purpose is to show what the all-atom distance divergence actually looks like in distribution space, and to verify that the quantitative JSD ranking in Figure~\ref{fig:jsd_fragment_main} corresponds to visible shifts in the intraligand distance profile.
Methods with larger JSD show broader distortions in the short- and mid-range distance modes, whereas PocketVE remains comparatively close to the reference distribution.

\begin{figure}[htbp]
\centering
\includegraphics[width=\textwidth]{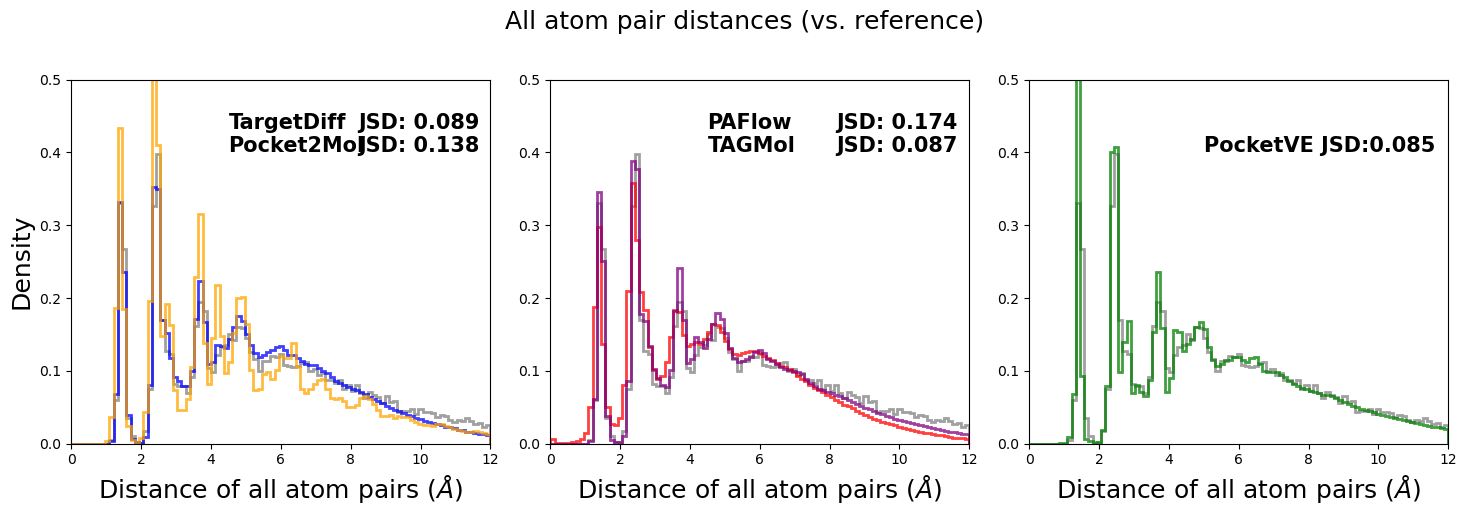}

\captionsetup{skip=4pt}
\caption{All-atom pair distance distributions compared against the reference ligands in the test set, following the same TargetDiff-style logic used for JSD diagnostics in the main text.
Each panel overlays the empirical distribution of reference intraligand atom-pair distances with the corresponding generated distribution, and reports the resulting Jensen--Shannon divergence.
The PocketVE panel corresponds to guided sampling at $s=5$.
Lower JSD indicates closer agreement with the reference distance profile, but this remains a distribution-fidelity diagnostic rather than an explicit physical-validity metric.}
\label{fig:appendix_jsd_overlay}

\end{figure}

\section{Test-Time Protein Perturbation}
\label{appendix:test_time_perturbation}

We use this experiment as a controlled stress test rather than as a standard evaluation setting.
At inference time, Gaussian noise is added only to the protein pocket coordinates, while ligand initialization, sampling schedule, and CFG scale are kept fixed at the main operating point $s=5$.
This setup probes whether training-time pocket perturbation improves tolerance to small errors in the conditioning geometry.

\begin{table}[htbp]
\centering

\captionsetup{skip=4pt}
\caption{Test-time protein perturbation stress test under inference-time pocket noise.
Gaussian noise with scale $\sigma_{\mathrm{pocket}}$ is added to pocket coordinates at inference time.
We compare PocketVE trained without protein perturbation, PocketVE trained with $\min(0.1\sigma,0.5)$ perturbation, and PocketVE trained with $\min(0.3\sigma,1.5)$ perturbation using the same four JSD diagnostics as in the main text together with median property metrics and key GenBench3D geometry metrics.
$\sigma_{\mathrm{pocket}}$ is measured in \AA{}.
Lower Local, 12\AA{}, CC-2\AA{}, Atom, Vina, Vina Min, Strain, and Centroid are better; higher QED, SA, Valid$_{3\text{D}}$, and Clash-Free are better.}
\label{tab:appendix_test_time_perturbation}
\scriptsize
\setlength{\tabcolsep}{2.9pt}
\resizebox{\linewidth}{!}{
\begin{tabular}{l|c|cccc|cccc|cccc}
\toprule
 & & \multicolumn{4}{c|}{JSD $(\downarrow)$} & \multicolumn{4}{c|}{Property / Affinity} & \multicolumn{4}{c}{Geometry} \\
\cmidrule(lr){3-6}\cmidrule(lr){7-10}\cmidrule(lr){11-14}
Method & $\sigma_{\mathrm{p}}$ & Local & 12\AA{} & CC-2\AA{} & Atom & QED $\uparrow$ & SA $\uparrow$ & Vina $\downarrow$ & \makecell{Vina\\Min $\downarrow$} & \makecell{Valid$_{3\text{D}}$\\$\uparrow$} & \makecell{Strain\\$\downarrow$} & \makecell{Clash-\\Free $\uparrow$} & \makecell{Centroid\\$\downarrow$} \\
\midrule
\multirow{5}{*}{\makecell[l]{PocketVE\\w/o perturb.}} & 0.0 & 0.265 & 0.0690 & 0.2750 & 0.1422 & 0.63 & 0.73 & -7.83 & -8.12 & 76.17 & 121.7 & 94.61 & 1.274 \\
 & 0.1 & 0.264 & 0.0672 & 0.2739 & 0.1419 & 0.65 & 0.73 & -7.81 & -8.15 & 79.48 & 127.7 & 94.76 & 1.301 \\
 & 0.3 & 0.269 & 0.0678 & 0.2754 & 0.1487 & 0.63 & 0.73 & -7.10 & -7.95 & 78.03 & 136.2 & 90.34 & 1.337 \\
 & 0.5 & 0.270 & 0.0713 & 0.2832 & 0.1552 & 0.64 & 0.74 & -5.90 & -7.58 & 79.28 & 130.6 & 81.14 & 1.403 \\
 & 1.0 & 0.275 & 0.0693 & 0.2802 & 0.1655 & 0.61 & 0.72 & -2.90 & -6.60 & 76.69 & 159.7 & 44.80 & 1.795 \\
\midrule
\multirow{5}{*}{\makecell[l]{PocketVE\\w/ $\min(0.1\sigma,0.5)$}} & 0.0 & 0.255 & 0.0646 & 0.2674 & 0.1225 & 0.65 & 0.71 & -7.87 & -8.18 & 80.05 & 128.4 & 93.79 & 1.250 \\
 & 0.1 & 0.263 & 0.0645 & 0.2689 & 0.1235 & 0.63 & 0.72 & -7.65 & -7.98 & 79.72 & 122.3 & 94.33 & 1.283 \\
 & 0.3 & 0.261 & 0.0640 & 0.2649 & 0.1261 & 0.64 & 0.73 & -6.82 & -7.65 & 78.43 & 117.6 & 91.26 & 1.281 \\
 & 0.5 & 0.268 & 0.0660 & 0.2745 & 0.1297 & 0.64 & 0.72 & -5.68 & -7.38 & 80.12 & 124.2 & 77.90 & 1.349 \\
 & 1.0 & 0.263 & 0.0629 & 0.2660 & 0.1293 & 0.63 & 0.70 & -2.75 & -6.63 & 80.56 & 149.1 & 41.26 & 1.651 \\
\midrule
\multirow{5}{*}{\makecell[l]{PocketVE\\w/ $\min(0.3\sigma,1.5)$}} & 0.0 & 0.265 & 0.0639 & 0.2637 & 0.1305 & 0.63 & 0.75 & -7.64 & -7.98 & 78.96 & 115.8 & 92.93 & 1.273 \\
 & 0.1 & 0.261 & 0.0652 & 0.2613 & 0.1340 & 0.65 & 0.76 & -7.48 & -7.87 & 80.68 & 104.7 & 94.32 & 1.294 \\
 & 0.3 & 0.266 & 0.0671 & 0.2664 & 0.1341 & 0.65 & 0.76 & -6.73 & -7.53 & 80.77 & 110.4 & 89.63 & 1.316 \\
 & 0.5 & 0.265 & 0.0653 & 0.2646 & 0.1393 & 0.65 & 0.77 & -5.78 & -7.36 & 81.15 & 110.5 & 78.95 & 1.353 \\
 & 1.0 & 0.266 & 0.0672 & 0.2658 & 0.1441 & 0.62 & 0.74 & -2.52 & -6.30 & 76.45 & 134.6 & 41.93 & 1.589 \\
\bottomrule
\end{tabular}
}

\end{table}

Across the sweep, both perturbation-trained variants generally preserve lower distance-distribution JSDs than the model trained without pocket perturbation.
The stronger $0.3\sigma$ training perturbation further improves SA and strain-energy behavior at several noise levels, while the milder $\min(0.1\sigma,0.5)$ training perturbation retains the strongest Vina-based scores at low test-time noise.
However, severe test-time corruption still substantially weakens Vina-based scores and clash-free rates for all three models.
These results suggest that protein perturbation mainly acts as a robustness regularizer for the conditioning geometry, rather than making the sampler invariant to arbitrarily corrupted pockets.

\section{Pocket-Diagnostic Details}
\label{appendix:pocket_diagnostics}

For pocket shuffling, we apply a fixed cyclic permutation with no self-matches across the 100 test pockets.
Each generated ligand is translated from its source pocket center to the target pocket center, then evaluated against the mismatched target.
Correct and shuffled conditions contain the same 9,471 ligands.
We first aggregate within each target pocket, average pockets equally, and form percentile 95\% CIs from 10,000 paired pocket-bootstrap replicates.
This is a post-hoc full-pipeline specificity control; it does not regenerate ligands while conditioning on an incorrect pocket and does not isolate a causal network component.

For PoseCheck, we evaluate original poses without redocking, summarize poses within each pocket, and average the 100 pockets equally.
Table~\ref{tab:posecheck_main} summarizes clashes, clash-free rates, and total interactions, and Table~\ref{tab:posecheck_categories} gives the full contact-category profile.
The changes are mixed: PocketVE does not maximize every interaction category, and the reference ligands themselves average 11.79 total interactions, close to PocketVE's 11.99.
Thus, total contact count should not be interpreted as interaction quality.

\begin{table}[htbp]
\centering

\captionsetup{skip=4pt}
\caption{Pocket-equal PoseCheck estimates on original poses with 10,000-replicate pocket-bootstrap 95\% CIs. PoseCheck clash-free denotes the percentage of poses with zero PoseCheck clashes and differs from the GenBench3D clash-free metric in Table~\ref{tab1}.}
\label{tab:posecheck_main}
\small
\setlength{\tabcolsep}{5pt}
\begin{tabular}{l|ccc}
\toprule
Method & Mean clashes $\downarrow$ & Clash-free (\%) $\uparrow$ & Mean interactions \\
\midrule
PocketVE & \textbf{8.52 [7.12, 10.10]} & \textbf{4.54 [2.91, 6.60]} & 11.99 [10.98, 13.01] \\
TargetDiff & 12.95 [11.19, 14.89] & 1.65 [1.04, 2.33] & 13.14 [12.15, 14.15] \\
PAFlow & 10.36 [7.87, 14.14] & 3.19 [1.88, 4.73] & 11.65 [10.78, 12.53] \\
TAGMol & 13.88 [11.81, 16.17] & 0.91 [0.31, 1.61] & 12.76 [11.83, 13.72] \\
Reference & 7.79 [6.49, 9.14] & 5.00 [1.00, 10.00] & 11.79 [10.63, 12.99] \\
\bottomrule
\end{tabular}
\end{table}

\begin{table}[htbp]
\centering

\captionsetup{skip=4pt}
\caption{Pocket-equal PoseCheck contact-category estimates with 10,000-replicate pocket-bootstrap 95\% CIs.}
\label{tab:posecheck_categories}
\scriptsize
\setlength{\tabcolsep}{3.5pt}
\begin{tabular}{l|rrrrrr}
\toprule
Method & Total interactions & HBD & HBA & Hydrophobic & VdW & Mean clashes \\
\midrule
PocketVE & 11.99 [10.98, 13.01] & 0.36 [0.30, 0.42] & 1.41 [1.14, 1.68] & 1.06 [0.88, 1.25] & 9.16 [8.47, 9.89] & 8.52 [7.12, 10.10] \\
TargetDiff & 13.14 [12.15, 14.15] & 0.68 [0.60, 0.76] & 1.63 [1.37, 1.90] & 1.15 [0.97, 1.35] & 9.69 [9.02, 10.38] & 12.95 [11.19, 14.89] \\
PAFlow & 11.65 [10.78, 12.53] & 0.39 [0.33, 0.45] & 1.23 [1.00, 1.47] & 1.59 [1.36, 1.83] & 8.44 [7.86, 9.05] & 10.36 [7.87, 14.14] \\
TAGMol & 12.76 [11.83, 13.72] & 0.59 [0.51, 0.68] & 1.50 [1.26, 1.76] & 1.17 [0.97, 1.38] & 9.50 [8.87, 10.15] & 13.88 [11.81, 16.17] \\
Reference & 11.79 [10.63, 12.99] & 0.50 [0.36, 0.65] & 1.96 [1.57, 2.37] & 0.72 [0.49, 0.98] & 8.61 [7.86, 9.38] & 7.79 [6.49, 9.14] \\
\bottomrule
\end{tabular}
\end{table}

Table~\ref{tab:posecheck_paired} reports paired differences with the sign convention PocketVE minus comparator.
Negative clash differences and positive clash-free differences favor PocketVE; interaction-count differences are descriptive because larger is not uniformly better.

\begin{table}[htbp]
\centering

\captionsetup{skip=4pt}
\caption{Paired pocket-level PoseCheck differences with percentile 95\% CIs. Clash-free differences are percentage points.}
\label{tab:posecheck_paired}
\small
\setlength{\tabcolsep}{4pt}
\begin{tabular}{l|rrr}
\toprule
Comparator & \(\Delta\) mean clashes & \(\Delta\) clash-free & \(\Delta\) total interactions \\
\midrule
TargetDiff & -4.42 [-5.92, -3.28] & +2.90 [+1.61, +4.64] & -1.15 [-1.47, -0.85] \\
PAFlow & -1.84 [-5.47, +0.63] & +1.35 [-0.30, +3.40] & +0.34 [-0.18, +0.88] \\
TAGMol & -5.36 [-7.26, -3.78] & +3.63 [+2.05, +5.67] & -0.77 [-1.25, -0.24] \\
Reference & +0.73 [-0.60, +2.17] & -0.46 [-4.37, +2.79] & +0.20 [-0.34, +0.73] \\
\bottomrule
\end{tabular}
\end{table}

\section{Property-Steering Diagnostics}
\label{appendix:property_steering}

We evaluate non-default and conflicting condition vectors with the same fixed PocketVE checkpoint, \(s=5\), 100 pockets, and 10 samples per pocket.
The condition order is (Vina, QED, SA), and larger indices denote more favorable training bins.
Table~\ref{tab:requested_condition_vectors} shows that ordered requests move all three aggregate medians in the requested directions.
Conflicting requests also expose cross-property competition: for example, \((4,0,2)\) improves Vina but lowers QED relative to \((2,4,2)\), while \((2,0,4)\) and \((2,4,0)\) reverse the QED--SA preference at similar Vina levels.
The non-monotonic outcomes preclude an interpretation as calibrated independent control.

\begin{table}[htbp]
\centering

\captionsetup{skip=4pt}
\caption{Aggregate property medians under ordered and conflicting requested bins. All rows use the same checkpoint and \(s=5\).}
\label{tab:requested_condition_vectors}
\small
\setlength{\tabcolsep}{8pt}
\begin{tabular}{c|ccc}
\toprule
Requested bins (Vina, QED, SA) & Vina Score median & QED median & SA median \\
\midrule
\((0,0,0)\) & -3.294 & 0.160 & 0.520 \\
\((1,1,1)\) & -5.177 & 0.284 & 0.540 \\
\((2,2,2)\) & -7.099 & 0.589 & 0.660 \\
\((3,3,3)\) & -7.288 & 0.619 & 0.690 \\
\((4,4,4)\) & -7.881 & 0.642 & 0.710 \\
\midrule
\((4,0,2)\) & -7.515 & 0.604 & 0.700 \\
\((2,4,2)\) & -7.144 & 0.670 & 0.650 \\
\((2,0,4)\) & -6.937 & 0.392 & 0.690 \\
\((2,4,0)\) & -6.865 & 0.605 & 0.530 \\
\bottomrule
\end{tabular}
\end{table}

Table~\ref{tab:cfg_property_attainment} reports distribution-level attainment under favorable joint requests.
Rates use successfully evaluated molecules as the denominator.
The Vina column uses \texttt{vina\_score}, whereas training boundaries use \texttt{vina\_dock}; consequently, we report the score shift but do not treat it as strict Vina-bin attainment.

\begin{table}[htbp]
\centering

\captionsetup{skip=4pt}
\caption{Property-distribution diagnostics under CFG. QED+SA top requires both favorable property bins; all three top additionally reports the more stringent three-way criterion using the available score-only Vina boundary. Vina values are interpreted directionally because the evaluation and training metric variants differ.}
\label{tab:cfg_property_attainment}
\small
\setlength{\tabcolsep}{5pt}
\begin{tabular}{c|r|rrrrr}
\toprule
Scale & Evaluated \(n\) & QED top & SA top & QED+SA top & All three top & Vina median \\
\midrule
0 & 935 & 14.01\% & 7.06\% & 1.28\% & 0.00\% & -6.704 \\
1 & 961 & 22.58\% & 12.59\% & 3.02\% & 0.00\% & -7.196 \\
5 & 953 & 27.91\% & 19.31\% & 3.99\% & 0.10\% & -7.892 \\
10 & 925 & \textbf{30.59\%} & \textbf{21.51\%} & \textbf{5.51\%} & \textbf{0.11\%} & \textbf{-8.124} \\
\bottomrule
\end{tabular}
\end{table}


\section{CFG Diversity and Reference Coverage}
\label{appendix:cfg_diversity}

Table~\ref{tab:cfg_diversity} aggregates successfully evaluated molecules across all pockets.
Higher CFG mildly reduces molecule-level uniqueness and increases fingerprint similarity, but scaffold entropy and top-scaffold concentration do not deteriorate monotonically.
Reference-scaffold coverage decreases with guidance.
Because these summaries pool molecules across pockets, they characterize global chemical-space behavior and can mask within-pocket convergence or pocket-specific reference recovery.

\begin{table}[htbp]
\centering

\captionsetup{skip=4pt}
\caption{Global diversity and reference-coverage diagnostics under CFG. Reference coverage is the percentage of reference scaffolds recovered in the generated set.}
\label{tab:cfg_diversity}
\small
\setlength{\tabcolsep}{5pt}
\begin{tabular}{c|cc|ccc|cc}
\toprule
 & \multicolumn{2}{c|}{Uniqueness} & \multicolumn{3}{c|}{Scaffold} & \multicolumn{2}{c}{Similarity / Coverage} \\
\cmidrule(lr){2-3}\cmidrule(lr){4-6}\cmidrule(lr){7-8}
Scale & Molecule & Scaffold & Entropy & Top-1 & Top-10 & Pairwise Tan. & Ref. Coverage \\
\midrule
0 & \textbf{99.36\%} & 83.42\% & 0.9481 & 6.95\% & 15.19\% & \textbf{0.0919} & \textbf{13.70\%} \\
1 & 98.65\% & 86.16\% & 0.9616 & \textbf{5.31\%} & 12.59\% & 0.1072 & 12.33\% \\
5 & 97.27\% & \textbf{86.78\%} & \textbf{0.9636} & 5.35\% & \textbf{12.17\%} & 0.1176 & 9.59\% \\
10 & 96.97\% & 84.22\% & 0.9563 & 6.38\% & 13.41\% & 0.1184 & 8.22\% \\
\bottomrule
\end{tabular}
\end{table}

\section{Coordinate-Scale Robustness and Diagnostics}
\label{appendix:coordinate_scale_diagnostics}

We evaluate \(\sigma_{\mathrm{data}}\in\{5,7.5,10,12.5,15\}\) on the same 100 test pockets with 100 generated ligands per pocket, CFG scale \(s=5\), and the generalized-arcsin scheduler.
Accordingly, Table~\ref{tab:sigma_data_sweep} is a robustness sensitivity analysis over the evaluated scale range rather than a sharply optimized ranking.

\begin{table}[htbp]
\centering

\captionsetup{skip=4pt}
\caption{\(\sigma_{\mathrm{data}}\) sweep (100 pockets \(\times\) 100 ligands). Brackets are pocket-bootstrap 95\% CIs for metrics recoverable per pocket; QED, SA, and Vina are evaluator-level point estimates.}
\label{tab:sigma_data_sweep}
\scriptsize
\setlength{\tabcolsep}{2.6pt}
\begin{tabular}{c|c|ccc|ccc}
\toprule
\(\sigma_{\mathrm{data}}\) & Eval. success [CI] & QED & SA & Vina Score & Valid$_{3\text{D}}$ [CI] & Strain median [CI] & Clash-free [CI] \\
\midrule
5.0 & 93.79 [92.11, 95.25] & 0.614 & 0.736 & -7.453 & 79.01 [77.31, 80.73] & 126.46 [111.74, 137.27] & 94.47 [90.63, 97.62] \\
7.5 & 93.33 [91.72, 94.78] & 0.606 & 0.764 & -7.216 & 80.65 [78.91, 82.36] & 102.83 [91.59, 113.09] & 93.72 [89.65, 96.96] \\
10.0 & 94.55 [93.24, 95.68] & 0.613 & 0.725 & -7.370 & 79.48 [77.75, 81.12] & 126.16 [113.83, 139.25] & 94.01 [90.05, 97.26] \\
12.5 & 93.47 [92.03, 94.74] & 0.608 & 0.736 & -7.430 & 78.66 [77.01, 80.35] & 128.86 [114.42, 143.01] & 94.23 [90.20, 97.37] \\
15.0 & 91.85 [90.42, 93.21] & 0.597 & 0.748 & -7.306 & 81.73 [80.13, 83.27] & 125.58 [112.58, 139.49] & 93.70 [89.40, 97.10] \\
\bottomrule
\end{tabular}
\end{table}

No setting dominates all metrics: \(\sigma_{\mathrm{data}}=10\) has the highest evaluation success, \(5\) has the most favorable mean Vina score, \(7.5\) has the highest SA and lowest median strain, and \(15\) has the lowest evaluation success.
Distributional metrics are similarly stable and non-monotonic: JSD-All-12\AA{} ranges from 0.0592 to 0.0665, JSD-CC-2\AA{} from 0.2533 to 0.2819, and atom-type JS from 0.1215 to 0.1458.
We therefore use \(\sigma_{\mathrm{data}}=10\) as a fixed operating choice within a robust range, not as a sharply optimized constant.

Table~\ref{tab:backbone_sampler_diagnostics} complements the scale sweep with targeted backbone and coordinate-scale diagnostics.
The sampler-control result is reported with the component-wise ablation in Appendix~\ref{appendix:component_ablation}; removing sampling-time noise sharply lowers evaluation success and geometric quality, indicating that the reverse-process noise treatment is an important part of the integrated sampler.
The VEDA-style row retains the TAGMol architecture but uses \(\sigma_{\mathrm{data}}=1\) and \(s=0\), so it tests the bundled formulation rather than isolating \(\sigma_{\mathrm{data}}\).
Removing the \(\sigma_{\mathrm{data}}\) factor in Eq.~\eqref{eq:appendix_scaled_input} leaves reconstruction and ligand-only geometry metrics similar but destroys pocket-aware spatial compatibility, as reflected by pathological Vina scores and a collapse in clash-free poses.
The no-Eq.~\eqref{eq:appendix_scaled_input} row is interpreted as a targeted compatibility diagnostic rather than as an isolated causal estimate of all geometric gains.

\begin{table}[htbp]
\centering

\captionsetup{skip=4pt}
\caption{Backbone and coordinate-scale diagnostics ($100$ pockets \(\times\) $10$ ligands). QED, SA, and Vina are means; strain is the pooled median. The VEDA-style and no-Eq.~\eqref{eq:appendix_scaled_input} rows are targeted bundled diagnostics rather than isolated causal comparisons.}
\label{tab:backbone_sampler_diagnostics}
\scriptsize
\setlength{\tabcolsep}{2.8pt}
\begin{tabular}{l|rrrrrrrr}
\toprule
Setting & Eval. success & QED & SA & Vina & Valid$_{3\text{D}}$ & Strain & Clash-free & Centroid (\AA) \\
\midrule
PocketVE setup & 95.30 & 0.624 & 0.722 & -7.589 & 78.72 & 127.94 & 94.23 & 1.25 \\
VEDA-style ($\sigma_{\mathrm{data}}=1$, $s=0$) & 89.50 & 0.400 & 0.563 & -5.454 & 43.43 & 1007.53 & 91.40 & 1.19 \\
Without scale preservation & 97.50 & 0.628 & 0.785 & +92.237 & 82.39 & 100.41 & 3.39 & 3.33 \\
\bottomrule
\end{tabular}
\end{table}

\section{Additional Case Studies}
\label{appendix:case_studies}

Figure~\ref{fig:additional_case_studies} provides three further case studies.
It compares the reference pose with representative baseline generations and the PocketVE sample for specific binding pockets.
These examples are included to show that the qualitative trends discussed in the main text are not limited to a single target.

\begin{figure}[htbp]
\centering
\includegraphics[width=\textwidth]{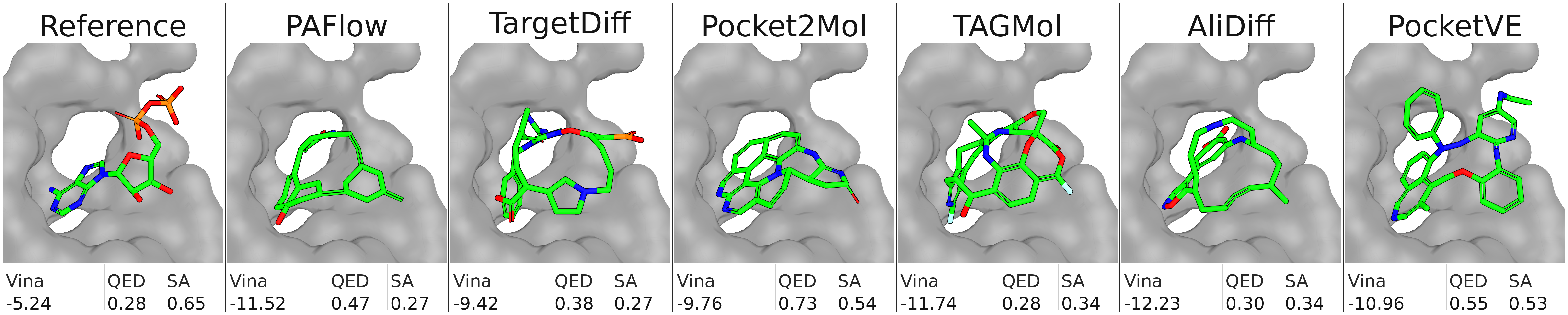}

\vspace{0.5em}
\includegraphics[width=\textwidth]{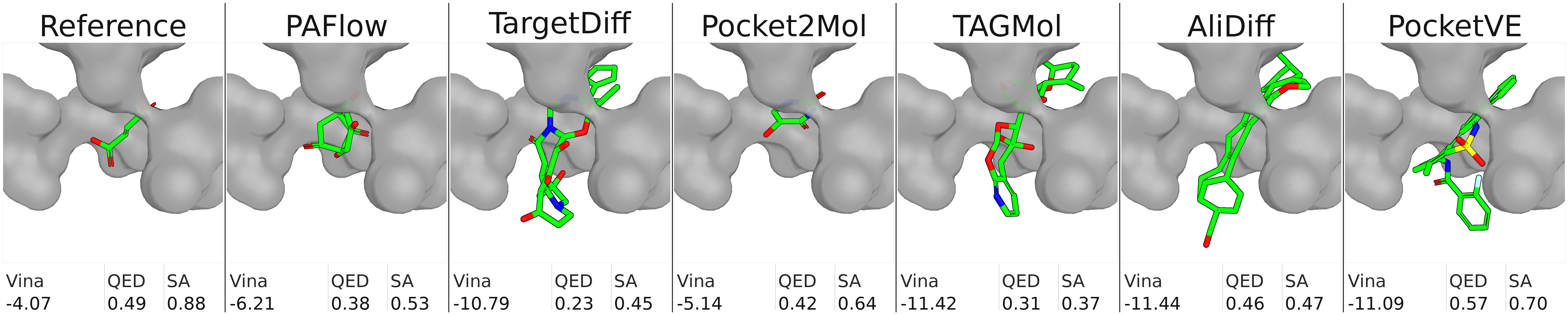}

\vspace{0.5em}
\includegraphics[width=\textwidth]{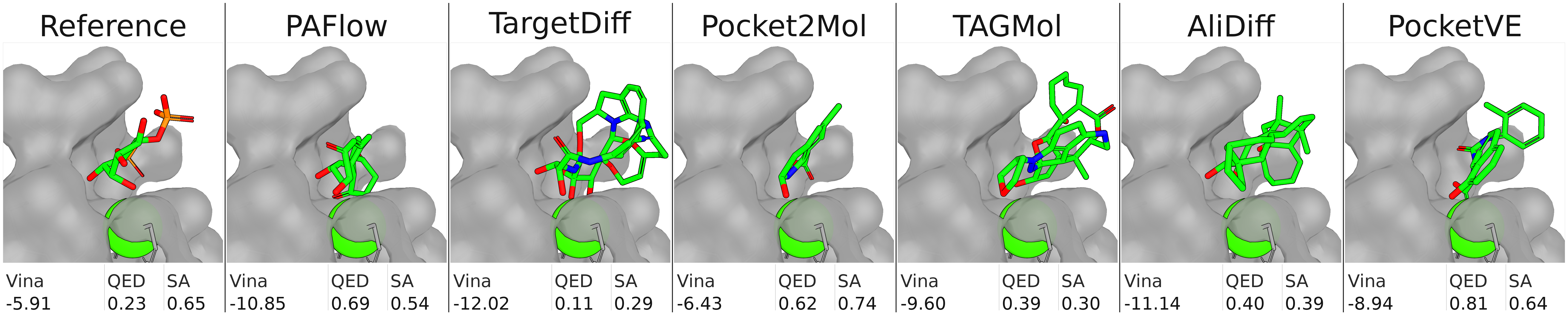}

\captionsetup{skip=4pt}
\caption{Additional case studies for PDB ID 5W2G (top), 3W83 (middle), and 1DJY (bottom).}
\label{fig:additional_case_studies}

\end{figure}



\end{document}